\documentclass[journal]{IEEEtranTIE}
\usepackage{graphicx}
\usepackage{cite}
\usepackage{picinpar}
\usepackage{amsmath}
\usepackage{amssymb}
\usepackage{amsthm}
\usepackage{url}
\usepackage{flushend}
\usepackage{colortbl}
\usepackage{soul}
\usepackage{multirow}
\usepackage{pifont}
\usepackage{color}
\usepackage{alltt}
\usepackage[hidelinks]{hyperref}
\usepackage{enumerate}
\usepackage{siunitx}
\usepackage{breakurl}
\usepackage{epstopdf}
\usepackage{pbox}
\usepackage{orcidlink}
\usepackage[utf8]{inputenc}
\usepackage{verbatim}
\usepackage{makecell}

\usepackage{amsmath, amssymb, amsthm}
\usepackage{mathtools}
\usepackage{bm}

\usepackage{amsthm}
\theoremstyle{plain}

\newtheorem{theorem}{Theorem}
\newtheorem{lemma}{Lemma}
\newtheorem*{remark}{Remark}
\newtheorem{assumption}{Assumption}

\theoremstyle{definition}

\begin{document}
\title{MAPS: A Mode-Aware Probabilistic Scheduling Framework for LPV-Based Adaptive Control\thanks{This work has been submitted to the IEEE for possible publication. Copyright may be transferred without notice, after which this version may no longer be accessible.}}
\author{
	\vskip 1em
	Taehun Kim,
    Guntae Kim,
	Cheolmin Jeong,
	and Chang~Mook~Kang
}

\maketitle

\begin{abstract}
    This paper proposes Mode-Aware Probabilistic Scheduling (MAPS), a novel adaptive control framework tailored for DC motor systems experiencing varying friction. MAPS uniquely integrates an Interacting Multiple Model (IMM) estimator with a Linear Parameter-Varying (LPV) based control strategy, leveraging real-time mode probability estimates to perform probabilistic gain scheduling. A key innovation of MAPS lies in directly using the updated mode probabilities as the interpolation weights for online gain synthesis in the LPV controller, thereby tightly coupling state estimation with adaptive control. This seamless integration enables the controller to dynamically adapt control gains in real time, effectively responding to changes in frictional operating modes without requiring explicit friction model identification. Validation on a Hardware-in-the-Loop Simulation (HILS) environment demonstrates that MAPS significantly enhances both state estimation accuracy and reference tracking performance compared to Linear Quadratic Regulator (LQR) controllers relying on predefined scheduling variables. These results establish MAPS as a robust, generalizable solution for friction-aware adaptive control in uncertain, time-varying environments, with practical real-time applicability.
\end{abstract}

\begin{IEEEkeywords}
Interacting multiple model, linear time varying, gain scheduling, adaptive control
\end{IEEEkeywords}

\definecolor{limegreen}{rgb}{0.2, 0.8, 0.2}
\definecolor{forestgreen}{rgb}{0.13, 0.55, 0.13}
\definecolor{greenhtml}{rgb}{0.0, 0.5, 0.0}

\section{Introduction}
\IEEEPARstart{D}{C} motors are widely employed in industrial automation, robotics, and automotive systems, where precise position and velocity control are essential for reliable performance. However, real-world applications of DC motors are often affected by dynamic uncertainties, such as time-varying friction, external disturbances, and parameter drifts, all of which challenge the integrity of conventional control strategies. In particular, variations in viscous friction coefficients arising from changes in temperature, mechanical wear, or lubrication conditions can introduce significant tracking errors, oscillations, and long-term performance degradation~\cite{ArmstrongHelouvry1994,Canudas1995, virgala2013friction}.

Friction in electromechanical systems typically consists of multiple nonlinear components, including Coulomb friction, viscous friction, and Stribeck effects. Traditional fixed-gain linear controllers such as PID and LQR, while effective under nominal conditions, often degrade in performance when friction characteristics deviate from the design assumptions~\cite{Yang1993,Bamieh1999}. Although sliding mode control, disturbance compensation, and adaptive control techniques have been introduced to address robustness against such uncertainties~\cite{Yao2013,Dinesh2016,Ren2019, Zhou2021}, challenges remain in achieving optimal performance under highly nonlinear conditions. Meanwhile, recent developments in friction modeling using machine learning and neural network have demonstrated promise in capturing complex, nonlinear behaviors~\cite{Khurram2015,Hu2019,Ren2021}, but they often lack interpretability and certification when applied to safety-critical systems.

\begin{figure}[!t]
    \centering
    \includegraphics[width=0.5\textwidth]{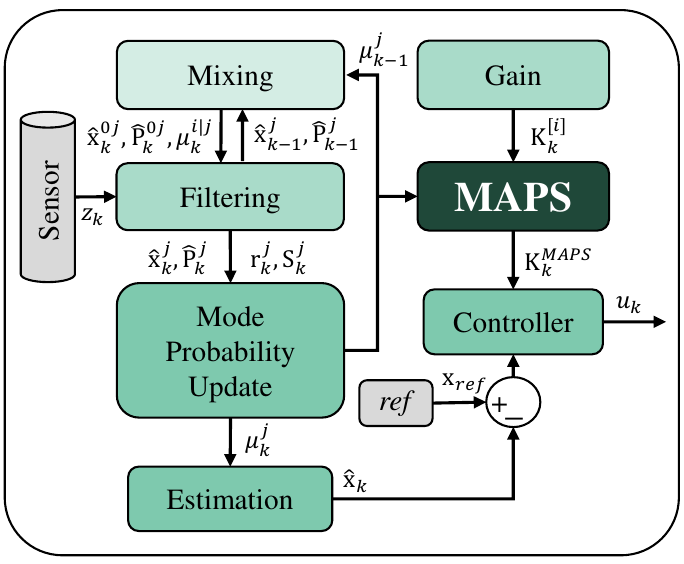}
    \caption{MAPS framework architecture.}
    \label{fig:archi}
\end{figure}

Recent years have witnessed significant advances in parameter estimation techniques tailored for adaptive control of nonlinear and time-varying systems. Methods such as adaptive observers, extended Kalman filter, and IMM algorithm have been extensively studied for their ability to identify varying system parameters and modes in real time~\cite{Blom1988, mahyuddin2013adaptive, reina2019vehicle}. Meanwhile, data-driven approaches employing machine learning and recursive regression have increasingly been leveraged to capture system dynamics without relying on explicit physical models~\cite{ding2016recursive, shlezinger2023model, kim2024data, lee2025bidirectional}. These estimation methodologies improve state and parameter reconstruction fidelity, especially under nonlinear and stochastic uncertainties~\cite{guo2014anti, chen2015disturbance}.

Despite these developments, parameter estimation techniques typically operate somewhat independently from controller design and adaptation. They provide state or parameter estimates which are subsequently fed into control algorithms designed separately. This architectural decoupling often introduces latency and limits the control system's ability to respond robustly and optimally to abrupt or latent changes in system dynamics issues that are particularly acute when the underlying scheduling variables are unmeasured or poorly observable, such as friction variations in electromechanical systems~\cite{guo2014anti}.

To overcome these limitations and construct an estimation control framework grounded in physical reality, our work experimentally identified the viscous friction coefficient through steady-state velocity measurements under varying load conditions in a HILS environment. By applying different external loads and measuring the corresponding steady-state speeds, a linear regression approach was used to estimate the realistic range of friction coefficients, defining the minimum and maximum bounds prevalent in the system. These practically obtained bounds enable the construction of vertex models that faithfully reflect frictional uncertainty, which serve as a robust foundation for our tightly integrated adaptive estimation and control design.

On the control front, gain scheduling methods based on LPV systems stand out for enabling online adaptation of controller gains by interpolating among vertex gains tied to measurable scheduling parameters~\cite{Bamieh1999, atoui2022toward}. Classical gain scheduling techniques, including lookup tables or linearization-based interpolation, enjoy industrial popularity due to their practical implementability~\cite{yahagi2022direct, hencey2009robust}. However, the effectiveness of these methods hinges critically on the choice and availability of accurate scheduling variables. Unmeasurable nonlinearities or rapidly switching dynamics often violate the assumptions needed for stable and efficient gain interpolation, leading to suboptimal or unstable performance~\cite{Kim2020, boufrioua2022gain, wei2014survey}.

Recent contributions have sought to enhance gain scheduling robustness through adaptive or polytopic control synthesis that partially accounts for parameter uncertainties~\cite{heemels2010observer, li2021polytopic}. Nevertheless, these approaches generally lack a principled integration with online estimation frameworks, thereby failing to exploit real-time probabilistic insights about system modes or latent parameters. The prevalent bifurcation of estimation and control adaptation domains impedes achieving seamless, robust, and optimal performance in systems subject to complex, latent, or abrupt uncertainties.

From a practical standpoint, integrated adaptive control schemes capable of probabilistically estimating system mode transitions and directly embedding this information into gain scheduling constitute a promising frontier. Such frameworks would enable controllers to dynamically and optimally respond to hidden dynamics, overcoming the functional disconnect that hampers most conventional approaches. This integration is especially critical in safety and performance critical applications, including advanced robotics and automotive control, where resilience against latent friction and load disturbances determines system reliability and efficiency~\cite{chu2022observer, zhang2014robust}.

Motivated by this critical gap, the present work introduces MAPS framework that tightly couples IMM-based mode estimation with LPV gain-scheduled control synthesis. By leveraging real-time mode probabilities as convex scheduling weights, MAPS realizes adaptive controller gains that respond smoothly and robustly to changing uncertainty modes without relying on explicit friction modeling or heuristic scheduling assumptions. 

The need for such reliable and adaptive control architectures is rapidly increasing as intelligent mobility, advanced manufacturing, and healthcare robotics demand resilience to unmeasurable or abrupt parameter variations for safety and efficiency. Unlike traditional methods, MAPS advances the field by uniting probabilistic inference with robust polytopic gain scheduling, achieving closed-loop performance that is not only theoretically grounded but also experimentally validated in realistic settings. While the present study demonstrates its effectiveness on friction-varying DC motor systems, the underlying methodology is inherently generalizable to a wide range of hybrid or parameter-varying systems with latent uncertainties. 

In MAPS, the probabilistic state and mode estimates from the MAPS are used to schedule optimal LQR gains dynamically, enabling friction-aware adaptive control without the need for explicit friction model identification. This approach compares favorably with LQR methods~\cite{Bamieh1999,Kang2018,Seo2022}, which rely on predefined scheduling variables, and ensures both real-time adaptability and optimal gain response to varying friction levels. In contrast to previous studies that primarily focus on simulation, our method is validated on a real HILS testbed using the QUBE-Servo 2 platform.

In summary, this work presents a unified framework that bridges robust mode estimation and optimal control adaptation for DC motors with uncertain and time-varying dynamics. The integration of probabilistic reasoning into the control loop introduces a new avenue for friction aware adaptive control that is both theoretically grounded and practically validated. The main contributions of this paper are summarized as follows:

\begin{itemize}
    \item We propose a novel framework that leverages the mode probabilities from an IMM estimator as scheduling weights in a LPV control system, enabling a seamless integration of discrete-mode state estimation with continuous control adaptation.

\item A gain synthesis method, termed MAPS-gain, is developed by aggregating multiple mode-specific LQR gains through probabilistic weighting. This approach maintains a linear controller structure while adaptively responding to nonlinear and time-varying system dynamics.

\item The effectiveness and real-time applicability of the proposed MAPS control framework are validated through both Simulink and HILS experiments on a DC motor platform, demonstrating superior performance compared to conventional gain scheduled and fixed-gain LQR controllers.
\end{itemize}

\section{System Modeling}
\subsection{DC Motor Model}
The QUBE-Servo 2 DC motor is selected as the experimental platform. This system consists of a brushed DC motor, high-resolution optical encoder, voltage amplifier, and integrated data acquisition unit. The electrical and mechanical dynamics are described as:
\begin{align}
L_m \frac{di}{dt} &= -R_m i - K_e \omega + v \\
J_{eq} \frac{d\omega}{dt} &= K_t i - b_m \omega
\end{align}
where $L_m$ is inductance, $R_m$ is resistance, $K_e$ is back-EMF constant, $K_t$ is torque constant, $J_{eq}$ is equivalent inertia, and $b_m$ is nominal viscous friction coefficient. The detailed definitions and numerical values of all parameters can be found in Table~\ref{tab:params}.

\subsection{State-Space Representation}
Choosing the angular position (\( \theta \)), angular velocity (\( \omega \)), and armature current (\( i \)) as the state variables, the state-space representation of the system is formulated as
\begin{equation}
    \mathrm{x} = \begin{bmatrix} \theta & \omega & i \end{bmatrix}^\top, \quad u = v
\end{equation}
where \( \mathrm{x} \in \mathbb{R}^n \) denotes the state vector, and \( u \in \mathbb{R}^m \) is the input voltage. Here, $n$ and $m$ represent the dimensions of the state and input vectors, respectively.

The system dynamics and output are expressed by
\begin{equation}
    \dot{\mathrm{x}} = \mathrm{A} \mathrm{x} + \mathrm{B}u, \quad \mathrm{y} = \mathrm{C} \mathrm{x}
    \label{eq:continuous_model}
\end{equation}
with system matrices defined as
\[
\mathrm{A} = \begin{bmatrix}
    0 & 1 & 0 \\
    0 & -\dfrac{b_m}{J_{eq}} & \dfrac{K_t}{J_{eq}} \\
    0 & -\dfrac{K_e}{L_m} & -\dfrac{R_m}{L_m}
\end{bmatrix}, \quad
\mathrm{B} = \begin{bmatrix}
    0 \\ 0 \\ \dfrac{1}{L_m}
\end{bmatrix}, \quad
\mathrm{C} = \begin{bmatrix}
    1 & 0 & 0
\end{bmatrix}
\]

The matrix \( \mathrm{A} \) represents the system dynamics among the state variables. The matrix \( \mathrm{B} \) describes how the control input affects the states. The matrix \( \mathrm{C} \) selects the output variable being measured, which is the angular position (\(\theta\)) in this case.
The physical parameters used in these matrices are summarized in Table~\ref{tab:params}.

\begin{table}[!ht]
\caption{QUBE-Servo 2 Model Parameters}
\label{tab:params}
\centering
\begin{tabular}{l l l}
\hline\hline
\multicolumn{1}{c}{Symbol} & \multicolumn{1}{c}{Quantity} & \multicolumn{1}{c}{\pbox{20cm}{Value}}  \\ \hline
$K_t$ & torque constant & 0.042 N$\cdot$m/A \\
$K_e$ & back-emf constant & 0.042 V$\cdot$s/rad \\\
$J_{r}$ & rotor inertia & $4.0 \times 10^{-6}$ kg$\cdot$m$^2$ \\
$J_{h}$ & hub inertia & $0.6 \times 10^{-6}$ kg$\cdot$m$^2$ \\
$J_{d}$ & disc moment of inertia & $1.6 \times 10^{-5}$ kg$\cdot$m$^2$ \\
$J_{eq}$ & equivalent moment of inertia & $J_{r} + J_{h} + J_{d} $ \\
$L_m$ & inductance & 1.16 mH \\
$R_m$ & resitance & 8.4 $\Omega$ \\
$b_m$ & viscous friction coefficient & $1.0 \times 10^{-5}$ N$\cdot$m$\cdot$s/rad \\
\hline\hline
\end{tabular}
\end{table}

\subsection{Viscous Friction Coefficient Identification}
\label{sec:friction_identification}

The viscous friction coefficient \( b \) was experimentally identified through steady-state velocity measurements under two distinct conditions: 
\begin{itemize}
    \item Minimum friction (without external load)
    \item Maximum friction (external load applied via finger pressure)
\end{itemize}


The friction torque in DC motors is modeled as the sum of static friction, Coulomb friction, and viscous friction~\cite{virgala2013friction, Wang2016, Keck2017, Piasek2019, Li2022}. The total friction torque $\tau_{\mathrm{fric}}$ is given by:
\begin{align}
    \tau_{\mathrm{fric}} = \tau_{\mathrm{static}} + \tau_{\mathrm{coulomb}} + \tau_{\mathrm{viscous}}
\end{align}

    \begin{align}
        \tau_{\mathrm{static}} =
        \begin{cases}
            \tau_s \cdot \mathrm{sgn}(\omega), & |\omega| = 0,\ |\tau_m| < \tau_s \\
            0, & \text{otherwise}
        \end{cases}
    \end{align}
    where $\tau_s$ is the static friction torque.
    \begin{align}
        \tau_{\mathrm{coulomb}} = \tau_c \cdot \mathrm{sgn}(\omega)
    \end{align}
    where $\tau_c$ is the Coulomb friction torque.
    \begin{align}
        \tau_{\mathrm{viscous}} = b \omega
    \end{align}
    where $b$ is the viscous friction coefficient.
Thus, the overall friction model used in this paper is:
\begin{align}
    \tau_{\mathrm{fric}} = \tau_s \cdot \mathrm{sgn}_0(\omega) + \tau_c \cdot \mathrm{sgn}(\omega) + b\omega
\end{align}
where $\mathrm{sgn}_0(\omega)$ denotes the sign function that is nonzero only at rest ($\omega = 0$ and $|\tau_m| < \tau_s$). To identify the viscous friction coefficient $b$, steady-state experiments were conducted by applying various voltages $V_m$ and measuring the corresponding steady-state angular velocities $\omega_m$. The following relationships were used:
\begin{align}
    T_m &= K_t i_m = b\omega_m \\
    V_m &= R_m i_m + K_e\omega_m
\end{align}
At steady-state ($\dot{\omega}_m = 0$), the current is:
\begin{align}
    i_m = \frac{b}{K_t}\omega_m
\end{align}
Substituting into the voltage equation:
\begin{align}
    V_m = R_m \left(\frac{b}{K_t}\omega_m\right) + K_e\omega_m = \omega_m \left( \frac{R_m b}{K_t} + K_e \right)
\end{align}
By linear regression of $V_m$ versus $\omega_m$, the slope $\mu$ is obtained:
\begin{align}
    V_m = \mu\omega_m \implies \mu = \frac{R_m b}{K_t} + K_e
\end{align}
Thus, the viscous friction coefficient is calculated as:
\begin{align}
    b = \frac{K_t}{R_m}(\mu - K_e)
\end{align}
where the slope $\mu$ is directly measured from the experimental results shown in Fig.~\ref{fig:friction_characteristics}.
\begin{figure}[!ht]
    \centering
    \includegraphics[width=0.5\textwidth]{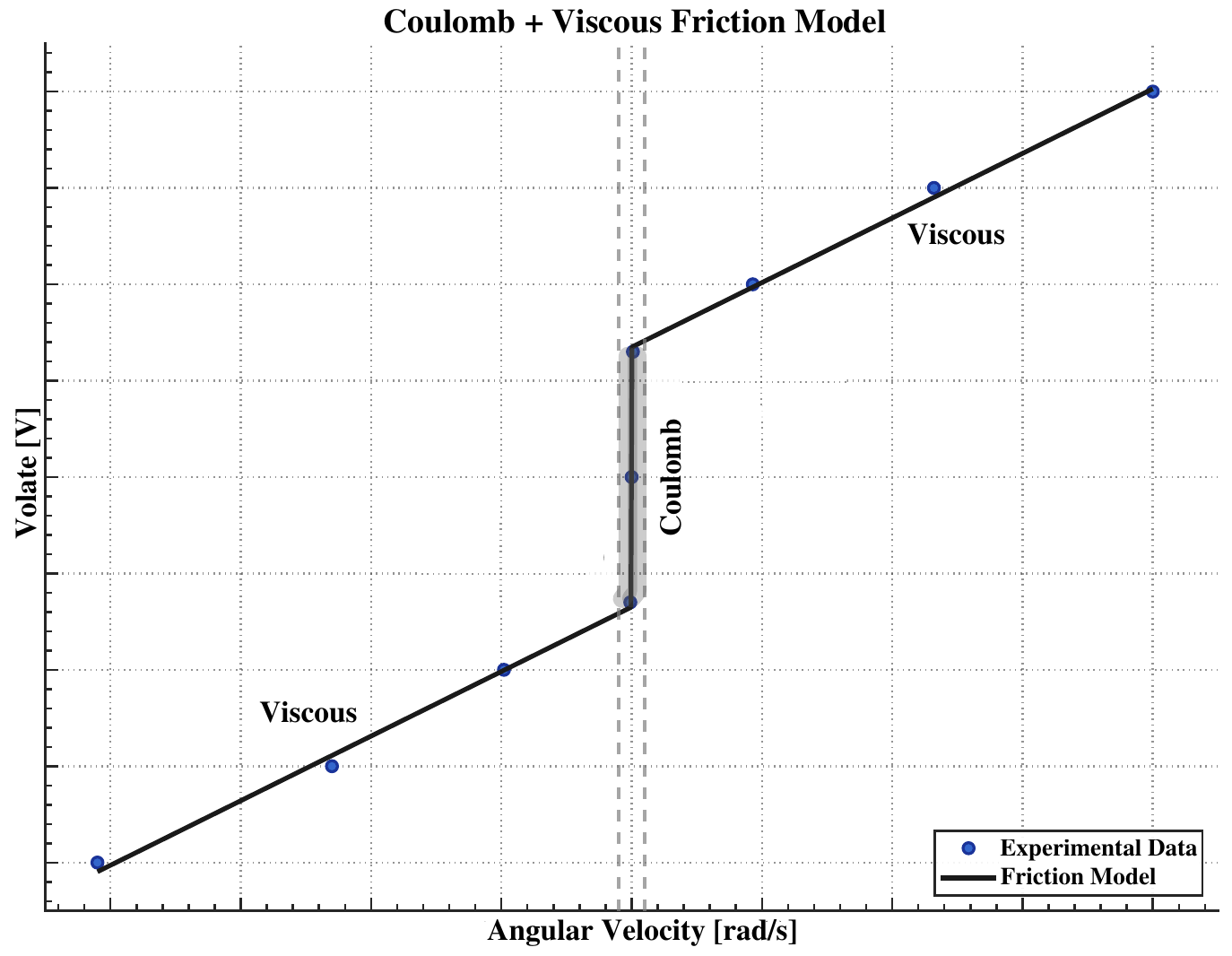}
    \caption{Friction characteristics showing coulomb/viscous regions.}
    \label{fig:friction_characteristics}
\end{figure}

Table~\ref{tab:min_friction} summarizes the measured steady-state velocities for varying input voltages under minimum friction conditions (no external load). The viscous friction coefficient $b_{min}$ was determined using linear regression of the voltage versus velocity data.

\begin{table}[!ht]
    \centering
    \caption{Steady-State Velocities (No External Load)}
    \label{tab:min_friction}
    \begin{tabular}{c|c}
        \hline\hline
        Voltage (V) & Velocity (rad/s) \\
        \hline
        -4 & $-97.7 \pm 0.3$ \\
        -3 & $-72.4 \pm 0.4$ \\
        -2 & $-45.9 \pm 0.2$ \\
        -1 & $-22.1 \pm 0.3$ \\
         1 & $22.5 \pm 0.4$ \\
         2 & $46.3 \pm 0.5$ \\
         3 & $72.4 \pm 0.6$ \\
         4 & $97.7 \pm 0.4$ \\
        \hline\hline
    \end{tabular}
\end{table}

With the slope $\mu_{min} = 0.042492$ obtained from linear regression, the viscous friction coefficient is calculated as
\[
b_{min} = \frac{0.042}{8.4}(\mu_{min} - 0.042) = 2.46 \times 10^{-6}\ \text{N}\cdot\text{m}\cdot\text{s}/\text{rad}
\]

Table~\ref{tab:max_friction} presents the steady-state velocities under maximum friction conditions, where an external load was applied. The slope from regression in this case is $\mu_{max} = 0.076$.

\begin{table}[!ht]
    \centering
    \caption{Steady-State Velocities (External Load Applied)}
    \label{tab:max_friction}
    \begin{tabular}{c|c}
        \hline\hline
        Voltage (V) & Velocity (rad/s) \\
        \hline
        -4 & $-41.0 \pm 1.2$ \\
        -3 & $-23.0 \pm 0.9$ \\
        -2 & $-9.8 \pm 0.7$ \\
         2 & $9.3 \pm 0.6$ \\
         3 & $23.2 \pm 1.1$ \\
         4 & $40.0 \pm 1.3$ \\
         \hline\hline
    \end{tabular}
\end{table}

The maximum friction coefficient is computed as
\[
b_{max} = \frac{0.042}{8.4}(\mu_{max} - 0.042) = 1.63 \times 10^{-4}\ \text{N}\cdot\text{m}\cdot\text{s}/\text{rad}
\]

In summary, the substantial difference between $b_{min}$ and $b_{max}$ highlights the need for adaptive control strategies suited to time-varying friction environments in practical DC motor systems.

\section{Filter Design}
\subsection{Simulation Modeling}
Although the system dynamics are originally described in continuous-time as shown in (\ref{eq:continuous_model}), most estimation and control algorithms are implemented in a discrete-time setting. This discretization is essential for compatibility with digital signal processors and real-time embedded hardware, while maintaining the fundamental behavior of the original system.

To this end, the continuous-time model (\ref{eq:continuous_model}) is discretized to obtain a discrete-time representation employed for both the Kalman filter and controller design, allowing for recursive computation at every sampling instance. In practical applications, modeling inaccuracies, measurement noise, and unexpected disturbances must also be accounted for. This leads to a standard discrete-time stochastic state-space model described as:

\begin{equation}
\begin{cases}
\mathrm{x}_{k+1} &= \Phi \mathrm{x}_k + \mathrm{\Gamma} u_k + \mathrm{w}_k, \label{eq:discrete_state}\\
\mathrm{y}_k &= \mathrm{H} \mathrm{x}_k + \mathrm{v}_k, 
\end{cases}
\end{equation}

where $k$ is the discrete-time index, and $\mathrm{w}_k$ and $\mathrm{v}_k$ represent process and measurement noise, respectively. Both are assumed to be zero-mean Gaussian random vectors with covariance matrices $\mathrm{Q}$ and $\mathrm{R}$, as specified in Section 3. Matrices $\mathrm{\Phi}$,$\mathrm{\Gamma}$,$\mathrm{H}$ are discretized by the forward Euler discretization~\cite{li2013discrete}. This provides a simple yet effective approximation of continuous-time system behavior within each sampling interval.

\subsection{Standard Kalman Filter}
In this section, the IMM filter and KF are constructed based on two distinct vertex models corresponding to the identified minimum and maximum friction values. The rationale and detailed structure of these vertex models are further elaborated in the subsequent LPV systems section. 

The KF is a recursive state estimator that offers optimal estimation for linear systems under the assumption of Gaussian noise~\cite{welch1995introduction, zeng2024state}. Although the system dynamics are initially modeled in continuous-time, the state estimation is carried out in discrete-time. Therefore, the discrete-time formulation of the KF is employed in this work. At each time step, the algorithm performs prediction and correction steps to recursively update the state and its covariance. At each time step, the KF performs prediction and update: \\

\textbf{Prediction:}
\begin{align}
\begin{cases}
\bar{\mathrm{x}}_{k|k-1} &= \mathrm{\Phi} \hat{\mathrm{x}}_{k-1|k-1} + \mathrm{\Gamma}u_k \\
\bar{\mathrm{P}}_{k|k-1} &= \mathrm{\Phi} \hat{\mathrm{P}}_{k-1|k-1} \mathrm{\Phi}^T + \mathrm{Q}
\end{cases}
\end{align}

\textbf{Update:}
\begin{align}
\begin{cases}
\mathrm{r}_k &= z_k - \mathrm{H} \bar{\mathrm{x}}_{k|k-1} \\
\mathrm{S}_k &= \mathrm{H}P_{k|k-1} \mathrm{H}^T + \mathrm{R} \\
\mathrm{K}_k &= \bar{\mathrm{\mathrm{P}}}_{k|k-1} \mathrm{H}^T \mathrm{S}_k^{-1} \\
\hat{\mathrm{x}}_{k|k} &= \hat{\mathrm{x}}_{k|k-1} + \mathrm{K}_k \mathrm{r}_k \\
\hat{\mathrm{P}}_{k|k} &= (\mathrm{I} - \mathrm{K}_k \mathrm{H}) \bar{\mathrm{P}}_{k|k-1} 
\end{cases}
\end{align}

Here, $\bar{\mathrm{x}}_{k|k-1}$ is the predicted state estimate at time step $k$ based on information available up to time step $k-1$ (prior estimate), and $\hat{\mathrm{x}}_{k|k}$ is the updated state estimate at time step $k$ obtained after incorporating the measurement $z_k$ (posterior estimate).

The measurement $z_k$ is acquired from the sensor. The process noise covariance matrix $\mathrm{Q}$ models uncertainties and unmodeled dynamics within the system, while the measurement noise covariance matrix $\mathrm{R}$ accounts for sensor noise and measurement inaccuracies. Choosing appropriate values for $\mathrm{Q}$ and $\mathrm{R}$ is essential: larger $\mathrm{Q}$ values increase the filter's responsiveness to model changes but can also amplify noise sensitivity; smaller $\mathrm{Q}$ values make the filter rely more on the model, reducing adaptability. Conversely, larger $\mathrm{R}$ values lead the filter to trust the model over measurements, while smaller $\mathrm{R}$ values heighten sensitivity to measurement noise~\cite{bavdekar2011identification, huang2017novel}.

\subsection{Interacting Multiple Model}
The IMM filter extends the KF by running multiple parallel KF, each with a different model(e.g. $b_{min}$ and $b_{max}$ for friction coefficient). The IMM algorithm adaptively fuses the state estimates using mode probabilities, enabling robust estimation under abrupt or gradual parameter changes. The IMM algorithm consists of the following steps~\cite{Blom1988, kirubarajan2004kalman,  Kim2020}: \\

\textbf{1. Mixing (Interaction):}
The IMM filter prepares the initial value for the next step by mixing the estimates between multiple modes at each time step. In this process, the previous state estimate and covariance are weighted averaged based on the mode transition probability. 

\begin{align}
\begin{cases}
\mu_{k|k-1}^j &= \sum_{i=1}^{N_v} \Pi_{ij} \mu_{k-1}^i \\
\mu_k^{i|j} &= \frac{\Pi_{ij} \mu_{k-1}^i}{\mu_{k|k-1}^j} \\
\hat{\mathrm{x}}_k^{0j} &= \sum_{i=1}^{N_v} \mu_k^{i|j} \hat{\mathrm{x}}_{k-1}^i \\
\hat{\mathrm{P}}_k^{0j} &= \sum_{i=1}^{N_v} \mu_k^{i|j} \left[ \hat{\mathrm{P}}_{k-1}^i + (\hat{\mathrm{x}}_{k-1}^i - \hat{\mathrm{x}}_k^{0j})(\hat{\mathrm{x}}_{k-1}^i - \hat{\mathrm{x}}_k^{0j})^T \right]
\end{cases}
\end{align} 

where $\Pi_{ij}$ denotes the Markovian mode transition probability from mode $i$ at the previous time step to mode $j$ at the current step $k$, and $\mu_{k-1}^i$ is the mode probability of mode $i$ at the previous step. Here, $N_v$ represents the total number of modes. The term $\mu_{k|k-1}^j$ represents the predicted probability of mode $j$ before incorporating the current measurement, while $\mu_k^{i|j}$ is the mixing probability that quantifies the contribution of mode $i$ to the mixed estimate for mode $j$. The mixed state estimate $\hat{\mathrm{x}}_k^{0j}$ and its associated covariance $\hat{\mathrm{P}}_k^{0j}$ serve as the initial conditions for the filtering operation under mode $j$. These quantities are computed as weighted averages over all feasible mode transitions, enabling consistent interaction among multiple models. \\

\textbf{2. Model-Conditioned Filtering:} Each model $j$ runs a standard KF using the mixed initial state $\hat{\mathrm{x}}_k^{0j}$ and covariance $\hat{\mathrm{P}}_k^{0j}$ from the interaction step. The filtering procedure for each model strictly follows the conventional prediction and update equations as defined in (19)-(25).\\

\textbf{3. Likelihood Calculation:} 
For each model, the innovation residual, its covariance, and the likelihood of the current measurement are computed to assess how well each model explains the observation.
\begin{align}
\begin{cases} 
\mathrm{r}_k^j &= z_k - \mathrm{H}^j \bar{\mathrm{x}}_k^j, \\
\mathrm{S}_k^j &= \mathrm{H}^j \bar{\mathrm{P}}_k^j (\mathrm{H}^j)^T + \mathrm{R},  \\
\Lambda_k^j &= e^{-\frac{1}{2} \left[ d \log(2\pi) + \log \det(\mathrm{S}_k^j) + (\mathrm{r}_k^j)^T (\mathrm{S}_k^j)^{-1} \mathrm{r}_k^j \right]}
\end{cases}
\end{align}
where $\mathrm{r}_k^j$ denotes the measurement residual for mode $j$, defined as the difference between the actual measurement $z_k$ and the predicted measurement $\mathrm{H}^j \bar{\mathrm{x}}_k^j$. The innovation covariance $\mathrm{S}_k^j$ combines the predicted estimation uncertainty $\bar{\mathrm{P}}_k^j$ with the measurement noise covariance $\mathrm{R}$, providing a measure of total uncertainty in the innovation. The scalar $\Lambda_k^j$ represents the likelihood of the measurement under mode $j$, computed assuming a Gaussian distribution of the residual. Here, $d$ is the dimension of the observation vector $\mathrm{y}_k$, and thus also the size of $\mathrm{S}_k^j$. The likelihood expression is derived from the Gaussian probability density function. \\

\textbf{4. Mode Probability Update:}
Using the likelihood $\Lambda^j_{k}$ calculated in the previous step, the mode probability $\mu^j_k$ for each model j is updated based on Bayes' rule, as shown in (21).
To enhance numerical stability and prevent underflow in likelihood calculations, the log-likelihood form is applied in practice.
\begin{equation}
\mu_k^j = \frac{\Lambda_k^j \mu_{k|k-1}^j}{\sum_{l=1}^{N_v} \Lambda_k^l \mu_{k|k-1}^l}
\end{equation}
Here, $\mu_k^j$ denotes the updated posterior mode probability of mode $j$ at time step $k$, incorporating the current measurement $y_k$ via the likelihood $\Lambda_k^j$. This term reflects how likely each mode is at the current step, after accounting for the observed data. In contrast to the one-step predicted mode probability $\mu_{k|k-1}^j$, which is computed purely from the Markov transition probabilities, the updated $\mu_k^j$ combines both the prior information and the current measurement, serving as the final probability in mode $j$ at time $k$. All mode probabilities $\mu_k^j$ satisfy $0 \leq \mu_k^j \leq 1$ and $\sum_{j=1}^n \mu_k^j = 1$.

\begin{remark}
The updated mode probabilities \( \mu^j_k \) are subsequently used as scheduling weights for online gain synthesis in the MAPS controller. This integration directly links state estimation with adaptive control, enabling the controller to respond in real time to dynamic mode transitions. 
\end{remark}

\textbf{5. Estimate Combination:}
Finally, the overall state estimate and covariance are calculated as the weighted sums across all models using the updated mode probabilities. 
\begin{align}
\begin{cases}
\hat{\mathrm{x}}_k &= \sum_{j=1}^{N_v} \mu_k^j \hat{\mathrm{x}}_k^j \\
\hat{\mathrm{P}}_k &= \sum_{j=1}^{N_v} \mu_k^j \left[ \hat{\mathrm{P}}_k^j + (\hat{\mathrm{x}}_k^j - \hat{\mathrm{x}}_k)(\hat{\mathrm{x}}_k^j - \hat{\mathrm{x}}_k)^T \right]
\end{cases}
\end{align}
In this process, $\hat{\mathrm{x}}_k$ represents the fused state estimate, obtained as the weighted average of the individual state estimates $\hat{\mathrm{x}}_k^j$ from each mode. The weights $\mu_k^j$ reflect the updated probability in each mode at time $k$, derived from the likelihood-informed Bayesian update. The total covariance $\hat{\mathrm{P}}_k$ accounts for both the within-mode uncertainty $\hat{\mathrm{P}}_k^j$ and the cross-mode deviation of each $\hat{\mathrm{x}}_k^j$ relative to the overall estimate $\hat{\mathrm{x}}_k$, capturing the overall estimation uncertainty in a statistically consistent manner. This final estimate serves as the unified output of the IMM filter at time $k$. 

The mode transition probability matrix $\Pi$ is tuned empirically. 

\begin{equation}
\Pi = \begin{bmatrix}
\Pi_{11} & \Pi_{12} & \cdots & \Pi_{1n} \\
\Pi_{21} & \Pi_{22} & \cdots & \Pi_{2n} \\
\vdots & \vdots & \ddots & \vdots \\
\Pi_{n1} & \Pi_{n2} & \cdots & \Pi_{nn}
\end{bmatrix}
\end{equation}
The main diagonal elements indicate a high probability of remaining in the same friction mode at each time step.

\section{Controller Design}
This section describes the design of the MAPS framework, which combines mode probability estimates from an IMM filter with gain scheduling in a LPV controller. By using these probabilities as interpolation weights for real-time state feedback gain synthesis, the proposed approach enables adaptive control under time-varying uncertainties such as friction. Theoretical stability of the closed-loop system is also addressed.

\subsection{Discrete-Time Linear Quadratic Regulator Theory}
The discrete-time LQR is a fundamental optimal control technique for linear systems in discrete-time.  
Its goal is to find a state feedback gain \( \mathrm{K} \) that minimizes the infinite-horizon quadratic cost function:
\begin{equation}
J = \sum_{k=0}^\infty \left( \mathrm{x}_k^T \mathrm{Q}_{LQR} \mathrm{x}_k + u_k^T \mathrm{R}_{LQR} u_k \right),
\end{equation}
where \( \mathrm{Q}_{LQR} \succeq 0 \) is a positive semi-definite state weighting matrix, and \( \mathrm{R}_{LQR} \succ 0 \) is a positive definite control weighting matrix.

The solution to the discrete-time LQR problem involves solving the discrete algebraic Riccati equation~\cite{liu2014adaptive, zhang2015lqr}:
\begin{equation}
\mathrm{P} = \mathrm{\Phi}^T \mathrm{P} \mathrm{\Phi} - \mathrm{\Phi}^T \mathrm{P} \mathrm{\Gamma} \left( \mathrm{R}_{LQR} + \mathrm{\Gamma}^T \mathrm{P} \mathrm{\Gamma} \right)^{-1} \mathrm{\Gamma}^T \mathrm{P} \mathrm{\Phi} + \mathrm{Q}_{LQR},
\end{equation}
where \( \mathrm{P} \) is the unique positive semi-definite solution.

The optimal state feedback control law minimizing \( J \) is given by:
\begin{equation}
u_k = \mathrm{K} \mathrm{x}_k, \quad \text{where} \quad \mathrm{K} = (\mathrm{R}_{LQR} + \mathrm{\Gamma}^T \mathrm{P} \mathrm{\Gamma})^{-1} \mathrm{\Gamma}^T \mathrm{P}\mathrm{\Phi}.
\end{equation}

It should be noted that, in this paper, the LQR gains were precomputed based on the nominal system dynamics without explicitly formulating the error-based (tracking) cost function~\cite{li2010torque,ruderman2008optimal}. This approach reflects a practical scenario in which the reference trajectory is unknown or varies arbitrarily during operation, making it difficult to design a dedicated error-based regulator a priori. Therefore, the controller was synthesized using state-based LQR gains, and the actual implementation utilizes the feedback of the reference error, \( \mathrm{e}_k = \mathrm{x}_{\mathrm{ref},k} - \hat{\mathrm{x}}_k \), to accommodate arbitrary reference commands. Empirically, this design choice was found to provide satisfactory tracking performance across a broad range of reference trajectories in HILS experiments.

For trajectory tracking or reference regulation purposes, the control input is often based on the error between the reference and the estimated system state, expressed as~\cite{gattami2009generalized, olalla2009robust}:
\begin{align}
\begin{cases}
\mathrm{e}_k = \mathrm{x}_{\mathrm{ref}, k} - \hat{\mathrm{x}}_k, \\
u_k = \mathrm{K} \mathrm{e}_k = \mathrm{K} (\mathrm{x}_{\mathrm{ref}, k} - \hat{\mathrm{x}}_k),
\end{cases}
\end{align}
where \( \mathrm{x}_{\mathrm{ref}, k} \in \mathbb{R}^n \) is the desired reference state at discrete-time \( k \), \( \hat{\mathrm{x}}_k \in \mathbb{R}^n \) is the estimated state obtained from an observer or estimator, and \( u_k \in \mathbb{R}^m \) is the control input applied to the plant. The matrix \( \mathrm{K} \in \mathbb{R}^{m \times n} \) is the discrete-time LQR gain designed to regulate the system based on the error \( \mathrm{e}_k \in \mathbb{R}^n\).

As in the continuous-time case, the choice of \( \mathrm{Q}_{LQR} \) and \( \mathrm{R}_{LQR} \) significantly influences the closed-loop system performance, stability, and robustness. A larger \( \mathrm{Q}_{LQR} \) penalizes deviations in the state more heavily, promoting aggressive regulation, whereas a larger \( \mathrm{R}_{LQR} \) penalizes control effort, yielding smoother but potentially slower responses.

\subsection{Linear Parameter Varying System}
A LPV system can be described as a convex combination of local linear time-invariant (LTI) models defined at the vertices of a parameter polytope~\cite{ Bamieh1999, White2016, Kang2018, Seo2022, Fenyes2023}. Let the scheduling parameter vector be $\rho_k \in \mathcal{P} \subset \mathbb{R}^{n_p}$, where $\mathcal{\mathrm{P}}$ is a convex polytope defined by its $N_v$ vertices $\{\rho^{[1]}, \ldots, \rho^{[N_v]}\}$. Here, $n_p$ denotes the dimension of the scheduling parameter vector, and $N_v$ is the number of vertices of the polytope. The state space representation is given by~\cite{heemels2010observer}:
\begin{align}
    {\mathrm{x}}_{k+1} &= \mathrm{\Phi}(\rho_k) \mathrm{x}_k + \mathrm{\Gamma}u_k  , & \mathrm{y}_k = \mathrm{H}\mathrm{x}_k,
    \label{eq:lpv_state}
\end{align}

where \( \mathrm{x}_k \in \mathbb{R}^{n} \) is the state vector, \( u_k \in \mathbb{R} \) is the input vector, and the output \( \mathrm{y}_k \in \mathbb{R} \) is a scalar.

In this paper, \(\mathrm{\Phi}\) denotes the discrete-time state transition matrix for the DC motor system. The \(\mathrm{\Phi}\) depends on the viscous friction coefficient, which varies in practical scenarios and introduces uncertainty into the system. To capture this variation, the system is modeled using a vertex model approach, where vertices correspond to system matrices constructed at experimentally identified extreme values of the viscous friction coefficient.

Specifically, the nominal system matrix \(\mathrm{\Phi}\) is defined using a representative nominal friction coefficient \(b_m\), representing the typical operating condition of the DC motor. However, since the friction varies over a range, two \textbf{vertex system matrices} \(\mathrm{\Phi}^{[1]}\) and \(\mathrm{\Phi}^{[2]}\) are constructed based on the identified minimum \(\underline{\rho} = b_{\min}\) and maximum \(\overline{\rho} = b_{\max}\) viscous friction coefficient vertices, respectively. These vertex models define the bounds of the system's operating range under parameter uncertainty. This polytopic representation ensures that any system behavior within the friction coefficient range can be accurately captured through a convex combination of these vertex system matrices.

The polytopic decomposition can be represented such that

\begin{align*}
\mathrm{\Phi}^{[1]} = \mathrm{\Phi}(0) + \underline{\rho}\hat{\mathrm{\Phi}}  \\
\mathrm{\Phi}^{[2]} = \mathrm{\Phi}(0) + \overline{\rho}\hat{\mathrm{\Phi}} 
\end{align*}

where
\begin{align*}
\Phi(0) =
\begin{bmatrix}
1 & T & 0 \\
0 & 1 & \frac{T \, K_t}{J_{eq}} \\
0 & -\frac{T \, K_e}{L_m} & 1 - \frac{T \, R_m}{L_m}
\end{bmatrix}
,\quad
\hat{\Phi} =
\begin{bmatrix}
0 & 0 & 0 \\
0 & -\frac{T \rho}{J_{eq} } & 0 \\
0 & 0 & 0
\end{bmatrix}
\end{align*}

where \( T = 0.002\, [s] \) is the fixed sampling time.  
\(\Phi(0)\) is the nominal system matrix representing the dynamics independent of the varying parameter \(\rho\).  
\(\hat{\Phi}\) is the fixed nodal matrix corresponding to the parameter-dependent variation associated with \(\rho\).  

The system matrix $\mathrm{\Phi}(\rho_k)$ are expressed as convex combinations of the matrices at the vertices:

\begin{equation}
\label{eq:lpv_phi}
\mathrm{\Phi}(\rho_k) = \sum_{i=1}^{N_v} \xi_i(\rho_k) \mathrm{\Phi}^{[i]},
\end{equation}

where $\mathrm{\Phi}^{[i]}$ are the system matrices at the $i$-th vertex, and the weighting functions $\xi_i(\rho_k)$ satisfy.

\begin{align}
    \xi_i(\rho_k) \geq 0, \quad \sum_{i=1}^{N_v} \xi_i(\rho_k) = 1,
\end{align}

The weighting functions $\xi_i(\rho_k)$ are typically determined by multilinear interpolation based on the current value of the scheduling parameter $\rho_k$ within the polytope. For efficient computation of the weighting functions $\xi_i(\rho_k)$ in LPV systems, we employ a matrix-based approach using barycentric coordinates~\cite{Kang2018}. 

Given the vertex matrix $\mathrm{V} \in \mathbb{R}^{(n_p+1) \times N_v}$ defined as
\begin{align}
    \mathrm{V} = \begin{bmatrix}
        \rho_1^{[1]} & \rho_1^{[2]} & \cdots & \rho_1^{[N_v]} \\
        \rho_2^{[1]} & \rho_2^{[2]} & \cdots & \rho_2^{[N_v]} \\
        \vdots & \vdots & \ddots & \vdots \\
        \rho_s^{[1]} & \rho_s^{[2]} & \cdots & \rho_s^{[N_v]} \\
        1 & 1 & \cdots & 1
    \end{bmatrix},
\end{align}

In this paper, the LPV system is represented as a convex combination of two models, corresponding to the minimum and maximum values of the scheduling parameter. The vertex matrix $\mathrm{V}$ in this case is given by:

\begin{equation}
\mathrm{V} = \begin{bmatrix}
    \overline{\rho} & \underline{\rho} \\
    1 & 1
\end{bmatrix} \in \mathbb{R}^{2 \times 2},
\end{equation}

the weighting function vector $\xi(\rho_k) = [\xi_1(\rho_k), \ldots, \xi_{N_v}(\rho_k)]^T$ is computed by solving the linear system

\begin{align}
    \mathrm{V}^T \xi(\rho_k) = \begin{bmatrix} \rho_k \\ 1 \end{bmatrix},
\end{align}

which yields

\begin{align}
    \xi(\rho_k) = (\mathrm{V}^T)^{-1} \begin{bmatrix} \rho_k \\ 1 \end{bmatrix}.
\end{align}

Vertex \( \mathrm{V} \) is uniquely defined based on the extreme values of the scheduling parameter, such as the identified minimum and maximum friction coefficients. Here, the index \( i \in \{1, \ldots, N_v\} \) corresponds to the two vertex models representing the lower and upper bounds of the parameter range, respectively. Since these values correspond to physical system bounds, the resulting vertex dynamics are fixed and known \textit{a priori}. This ensures that the associated LQR gains \( \mathrm{K}^{[i]} \) are consistently derived for a well-defined set of models within the LPV polytope.

\section{Closed-loop Stability Analysis of MAPS}

To formalize the closed-loop stability of the proposed MAPS-gain framework, we present two theorems based on the theory of discrete-time LPV systems. The first theorem ensures quadratic stability under the IMM-based gain scheduling strategy. The second theorem extends this guarantee by showing that exponential stability is preserved even under bounded parameter mismatch. Recall that the system matrix $\mathrm{\Phi}(\hat{\rho}_k)$ is constructed as a convex combination of vertex models as defined in (29), enabling the use of polytopic LPV stability tools.

\begin{lemma}[Quadratic Stability of Polytopic Vertex Systems]
Let $\mathrm{\Phi}^{[i]}$ be the vertex systems of a discrete-time polytopic LPV model. Assume that for each $i$, a stabilizing state feedback gain $K^{[i]}$ is designed such that the following inequality holds:
\begin{equation}
(\mathrm{\Phi}^{[i]} + \mathrm{\Gamma} \mathrm{K}^{[i]})^\top \mathrm{P} (\mathrm{\Phi}^{[i]} + \mathrm{\Gamma} \mathrm{K}^{[i]}) - \mathrm{P} \prec 0,\quad \forall i = 1, \ldots, N_v
\label{eq:common_lyap_discrete}
\end{equation}
for some symmetric positive definite matrix $P \succ 0$. Then, any convex combination of these systems, defined as
\begin{equation}
\mathrm{\Phi}^{\text{cl}}_{k} = \sum_{i=1}^{N_v} \mu_k^{(i)} (\mathrm{\Phi}^{[i]} + \mathrm{\Gamma} \mathrm{K}^{[i]}),\quad \sum_{i=1}^{N_v} \mu_k^{(i)} = 1,\; \mu_k^{(i)} \geq 0,
\end{equation}
is quadratically stable. Here, $\mathrm{\Phi}^{\text{cl}}_{k}$ denotes the convex combination of the time-varying closed-loop state transition matrices at discrete-time $k$.
\end{lemma}

The inequality in \eqref{eq:common_lyap_discrete} ensures that all vertex systems are individually quadratically stable under a common discrete-time Lyapunov function $V(\mathrm{x})= \mathrm{x}^\top \mathrm{P} \mathrm{x}$. This condition implies that any convex combination of the closed-loop systems such as $\Phi^{\text{cl}}_{k}$ is also stable for all admissible scheduling weights $\mu_k^{(i)}$ within the probability simplex.

\begin{proof}
This follows directly from convexity of quadratic forms and the discrete-time Lyapunov inequality. Since the inequality \eqref{eq:common_lyap_discrete} holds for all vertices and the weighting vector $\mu_k$ lies within the probability simplex $\mathcal{S}^{N_v}$, the convex combination also satisfies
\begin{equation}
(\mathrm{\Phi}^{\text{cl}}_{k})^\top \mathrm{P} \mathrm{\Phi}^{\text{cl}}_{k} - \mathrm{P} \prec 0
\end{equation}
for all admissible $\mu_k$, ensuring discrete-time quadratic stability.
\end{proof}

We consider that the parameter-dependent system and control matrices are modeled as convex combinations of vertex models. Specifically, the system state and gain at each time $k$ are expressed as weighted sums of vertex states and gains, where the weights $\mu_k^{(i)}$ which are the mode probabilities directly correspond to the states ${\xi}_i(\rho_k)$ of those vertices. In other words, the weighting factors and vertex states coincide exactly, yielding the parameter-dependent representation. Note that the expression in~(\ref{eq:lpv_phi}) can be equivalently represented as in~(\ref{eq:lpv_mu}).

\begin{equation}
\begin{aligned}
\label{eq:lpv_mu}
\mathrm{\Phi}(\rho_k) &= \sum_{i=1}^{N_v} \mu_k^{(i)} \mathrm{\Phi}^{[i]}, \\
\mathrm{K}(\rho_k) &= \sum_{i=1}^{N_v} \mu_k^{(i)} \mathrm{K}^{[i]},
\end{aligned}
\end{equation}
where \( \mu_k^{(i)} \in [0,1] \) are scheduling weights that satisfy
\[
\sum_{i=1}^{N_v} \mu_k^{(i)} = 1.
\]
This implies that the closed-loop matrix can also be expressed as a convex combination:
\begin{equation}
\mathrm{\Phi}_{\text{cl}}(\rho_k) = \mathrm{\Phi}(\rho_k) + \mathrm{\Gamma} \mathrm{K}(\rho_k) = \sum_{i=1}^{N_v} \mu_k^{(i)} (\mathrm{\Phi}^{[i]} + \mathrm{\Gamma} \mathrm{K}^{[i]}).
\end{equation}
Such a structure naturally fits within the polytopic LPV framework, allowing the use of a common Lyapunov function for stability analysis.

\begin{theorem}[Stability of MAPS-Gain Scheduled LPV Control]
Suppose that the mode probabilities $\mu_k^{(i)}$ at each discrete time step $k$ are generated by an IMM filter consisting of $N_v$ models. If each vertex system $\mathrm{\Phi}^{[i]}$ is stabilized by a gain $K^{[i]}$ satisfying the conditions in Lemma 1, then the MAPS-gain controller under perfect scheduling:
\begin{equation}
    \mathrm{K}({\rho}_k) = \sum_{i=1}^{N_v} \mu_k^{(i)} \mathrm{K}^{[i]}
\end{equation}
ensures quadratic stability of the closed-loop system:
\begin{equation}
    \mathrm{x}_{k+1} = (\mathrm{\Phi}({\rho}_k) + \mathrm{\Gamma} K({\rho}_k)) \mathrm{x}_k
\end{equation}
under all possible mode transitions modeled by the IMM.
\end{theorem}

\begin{proof}
Since the IMM mode probabilities $\mu_k$ remain within the convex simplex $\mathcal{S}^{N_v}$, and the feedback gain is a convex combination of stabilizing gains, the closed-loop system matrix $\mathrm{\Phi}^{\text{cl}}_k$ remains quadratically stable by Lemma 1.
\end{proof}

\begin{remark}
The MAPS-gain framework exploits a convex combination of precomputed stabilizing gains according to the IMM estimated mode probabilities. This avoids online Riccati computation and guarantees robust quadratic stability under mode switching and model uncertainty.
\end{remark}

However, since the parameter is estimated via the IMM filter, a mismatch may exist between the estimated value \( \hat{\rho} \) and the true parameter \( \rho \). Therefore, it is necessary to establish that the closed-loop system remains stable under the following assumptions, even in the presence of such estimation errors. In the proposed framework, the estimated scheduling parameter \( \hat{\rho}_k \) used for gain selection is obtained directly from the IMM mode probabilities. Specifically,
\begin{equation}
    \hat{\rho}_k := \sum_{i=1}^{N_v} \mu_k^{(i)} \rho^{[i]},
\end{equation}
where \( \rho^{[i]} \) is the representative scheduling parameter value associated with the $i$-th vertex model, and \( \mu_k^{(i)} \in [0,1] \) are the IMM-derived mode probabilities such that
\[
\sum_{i=1}^{N_v} \mu_k^{(i)} = 1.
\]
This expression links the probabilistic output of the IMM filter to the continuous scheduling variable required by the LPV framework.
To formally guarantee stability in the presence of scheduling mismatch, we introduce the following assumptions:
\begin{remark}
In the MAPS framework, both the system and control matrices are constructed as convex combinations of precomputed vertex models:
\[
\mathrm{\Phi}(\rho_k) = \sum_{i=1}^{N_v} \mu_k^{(i)} \mathrm{\Phi}^{[i]}, \quad
\mathrm{K}(\rho_k) = \sum_{i=1}^{N_v} \mu_k^{(i)} \mathrm{K}^{[i]},
\]
where the weights $\mu_k^{(i)}$ form a convex simplex. This structure ensures that $\mathrm{\Phi}(\cdot)$ and $K(\cdot)$ inherit Lipschitz continuity and stability properties from the vertex systems.
\end{remark}

\begin{assumption}[Bounded Estimation Error]
There exists $\epsilon > 0$ such that $\|\hat{\rho}_k - \rho_k\| < \epsilon$ for all $k$.
\end{assumption}

This assumption reflects that the scheduling parameter $\hat{\rho}_k$ estimated by the IMM filter has a uniformly bounded deviation from the true parameter $\rho_k$. 
In practical systems, the estimation error can be kept small through adequate sensor fusion and model calibration, making this condition realistic.

\begin{assumption}[Slowly Varying Parameters]
There exists $\delta > 0$ such that $\|\rho_{k+1} - \rho_k\| < \delta$ for all $k$.
\end{assumption}

This condition implies that the underlying physical parameter $\rho_k$ evolves gradually over time rather than changing abruptly. 
Such slow variations are common in systems like vehicle dynamics or battery state modeling, where the scheduling parameters (e.g., velocity, road friction) typically change smoothly.

\begin{assumption}[Lipschitz Continuity]
The mappings $\mathrm{\Phi}(\cdot)$ and $\mathrm{K}(\cdot)$ are Lipschitz continuous:
\begin{equation}
\begin{split}
\|\mathrm{\Phi}(\rho_1) - \mathrm{\Phi}(\rho_2)\| &\le L_\mathrm{\Phi} \|\rho_1 - \rho_2\|, \\
\|\mathrm{K}(\rho_1) - \mathrm{K}(\rho_2)\| &\le L_\mathrm{K} \|\rho_1 - \rho_2\|.
\end{split}
\end{equation}
\end{assumption}

Lipschitz continuity ensures that small changes in the scheduling parameter induce proportionally small variations in the system matrices and controller gains. 
This property is typically satisfied in parameter-dependent models derived via interpolation or offline gain scheduling, and is crucial for robust LPV control synthesis.

\begin{assumption}[Common Lyapunov Stability]
There exists $P \succ 0$ and $\alpha > 0$ such that
\[
\mathrm{\Phi}_{\text{cl}}(\rho)^\top \mathrm{P} \mathrm{\Phi}_{\text{cl}}(\rho) - \mathrm{P} \prec -\alpha \mathrm{I}, \quad \forall \rho,
\]
where $\mathrm{\Phi}_{\text{cl}}(\rho) := \mathrm{\Phi}(\rho) + \mathrm{\Gamma} \mathrm{K}(\rho)$.
\end{assumption}

This assumption asserts the existence of a common Lyapunov function for the entire family of closed-loop systems over the parameter space. 
Although conservative, it is widely used in polytopic LPV frameworks and is achievable by designing gains over a finite set of vertex models with convex stability guarantees.

\begin{theorem}[Exponential Stability of MAPS-Gain LPV Control under Parameter Estimation Error]
Define the closed-loop system matrix as:
\[
\mathrm{\Phi}_{\text{cl}}(\hat{\rho}_k) := \mathrm{\Phi}(\hat{\rho}_k) + \mathrm{\Gamma} \mathrm{K}(\hat{\rho}_k)
\]
Then the system evolves as:
\[
\mathrm{x}_{k+1} = \mathrm{\Phi}_{\text{cl}}(\hat{\rho}_k) \mathrm{x}_k
\]
Under Assumptions 1--4, consider the discrete-time LPV system
\[
    \mathrm{x}_{k+1} = \mathrm{\Phi}(\hat{\rho}_k)\,\mathrm{x}_k + \mathrm{\Gamma}\,\mathrm{K}(\hat{\rho}_k)\,\mathrm{x}_k,
\]
where $\hat{\rho}_k$ is an estimated scheduling parameter and the true value is $\rho_k$. 
There exist constants \( \epsilon^*>0 \) and \( \delta^*>0 \) such that if
\[
\|\hat{\rho}_k - \rho_k\| < \epsilon^*, \quad \|\rho_{k+1} - \rho_k\| < \delta^*, \quad \forall k,
\]
then the origin \( \mathrm{x} = 0 \) of the system is exponentially stable. That is, there exist constants \( C > 0 \) and \( 0 < \lambda < 1 \) such that
\[
\|\mathrm{x}_k\| \le C\,\lambda^k\,\|\mathrm{x}_0\|, \quad \forall k.
\]
\end{theorem}

\begin{proof}
Let $V(\mathrm{x}_k) = \mathrm{x}_k^\top \mathrm{P} \mathrm{x}_k$. For nominal scheduling ($\hat{\rho}_k = \rho_k$), Assumption 4 gives:
\[
V(\mathrm{x}_{k+1}) - V(\mathrm{x}_k) < -\alpha \|\mathrm{x}_k\|^2.
\]
Note that this bound accounts for both the second-order perturbation term \( \Delta_k^\top \mathrm{P} \Delta_k \) and the cross terms arising from mismatch. These cross terms are conservatively absorbed into the overall perturbation bound for analytical tractability.
For mismatched scheduling, define $\Delta_k := \mathrm{\Phi}_{\text{cl}}(\hat{\rho}_k) - \mathrm{\Phi}_{\text{cl}}(\rho_k)$. From Assumption 3,
\[
\|\Delta_k\| \le L_\mathrm{\Phi} \epsilon + \|\mathrm{\Gamma}\| L_\mathrm{K} \epsilon =: L \epsilon.
\]
Then,
\begin{equation}
\begin{aligned}
V(\mathrm{x}_{k+1}) &= 
\mathrm{x}_k^\top \mathrm{\Phi}_{\text{cl}}(\rho_k)^\top \mathrm{P} \mathrm{\Phi}_{\text{cl}}(\rho_k) \mathrm{x}_k \\
&\quad + \mathrm{x}_k^\top \Delta_k^\top \mathrm{P} \Delta_k \mathrm{x}_k \\
&\quad + 2 \mathrm{x}_k^\top \mathrm{\Phi}_{\text{cl}}(\rho_k)^\top \mathrm{P} \Delta_k \mathrm{x}_k,
\end{aligned}
\end{equation}

Bounding the perturbation and cross terms, we get:
\[
V(\mathrm{x}_{k+1}) \le V(\mathrm{x}_k) - \alpha \|\mathrm{x}_k\|^2 + \lambda_{\max}(\mathrm{P})\,L^2\,\epsilon^2 \|\mathrm{x}_k\|^2.
\]
Let $\tilde{\alpha} := \alpha - \lambda_{\max}(\mathrm{P})\,L^2\,\epsilon^2$. For 
\[
\epsilon < \epsilon^* := \sqrt{\frac{\alpha}{2\,\lambda_{\max}(\mathrm{P})\,L^2}},
\]
we have $\tilde{\alpha} > \alpha/2 > 0$ and
\[
V(\mathrm{x}_{k+1}) \le (1 - \tilde{\alpha}) V(\mathrm{x}_k).
\]
Using bounds on $V(\mathrm{x})$, we conclude:
\[
\|\mathrm{x}_k\| \le \sqrt{\frac{\lambda_{\max}(\mathrm{P})}{\lambda_{\min}(\mathrm{P})}} (1 - \tilde{\alpha})^{k/2} \|\mathrm{x}_0\| =: C \lambda^k \|\mathrm{x}_0\|.
\]
\end{proof}
\begin{remark}
While the proof focuses on the estimation error $\epsilon$, the bounded rate-of-change condition $\delta$ ensures that the parameter trajectory $\{\rho_k\}$ evolves smoothly over time. This prevents rapid fluctuations that could cause $\hat{\rho}_k$ to significantly diverge from $\rho_k$, helping maintain the small mismatch condition uniformly. Hence, both $\epsilon$ and $\delta$ jointly contribute to robust exponential stability.
\end{remark}

The MAPS-gain controller remains exponentially stable under bounded scheduling mismatch, validating its robustness for real-time implementation based on estimated parameters. This result bridges the gap between theoretical LPV control design and practical applications under imperfect parameter estimation, particularly relevant for systems utilizing real-time filters such as the IMM.

\section{Experimental Setup and Scenario}
This section outlines the experimental procedures adopted in this paper. First, simulation experiments in the Simulink environment were conducted to verify whether the proposed state estimation scheme can accurately infer time-varying friction coefficients. Based on these results, HILS experiments were performed to evaluate the integrated MAPS framework combining state estimation and controller in a real-time discrete-time setting, assessing its robustness and control performance under external load.

\subsection{Experimental Environments}
All real-time experiments were conducted using the QUBE-Servo 2 platform with MATLAB/Simulink 2020b and real-time interface. Each experiment had a duration of 30 seconds and was sampled at 500 Hz (2 ms sampling time). Angular position (\( \theta \)), angular velocity (\( \omega \)), current (\( i \)), and control input were recorded for evaluation. The standard KF and MAPS were implemented using empirically tuned noise covariances. The process noise covariance $\mathrm{Q}$ was diagonal with $1\times10^{-6}$ for each state variable, and the measurement noise covariance $\mathrm{R}$ was fixed at $1\times10^{-5}$. The IMM filter consisted of two friction models—representing minimum and maximum friction and used a mode transition probability matrix where the probability of remaining in the current mode was fixed at 0.9. For control design, an LQR controller was used with $\mathrm{Q}_{LQR} = \text{diag}(100, 1, 1)$ and $\mathrm{R}_{LQR} = 10$, prioritizing tight regulation of angular position (\( \theta \)) while moderately attenuating input effort. Two classes of experimental scenarios were investigated: (i) step input responses under both load and no-load conditions, and (ii) sine wave tracking in similar settings. These scenarios aim to analyze both transient and steady-state performance characteristics under realistic friction variations.

\subsection{Experimental Scenarios}
The experimental investigation in this study consists of two main scenarios. First, to evaluate the response of the system to a step input, experiments are performed under two conditions: with and without an additional load directly applied to the disk. The resulting behaviors are compared to characterize the influence of load on system dynamics. Second, for the sine wave tracking task, both load-free and loaded cases are considered, and in each scenario, the performance of the MAPS-gain LQR controller is compared against that of a fixed-gain LQR controller. This comparison is conducted to demonstrate the effectiveness of each control strategy under realistic, time-varying friction and load conditions. Through these scenarios, the integrated performance of state estimation and control is systematically evaluated under varying experimental conditions.

\begin{figure}[!ht]
  \centering
  \includegraphics[width=0.5\textwidth]{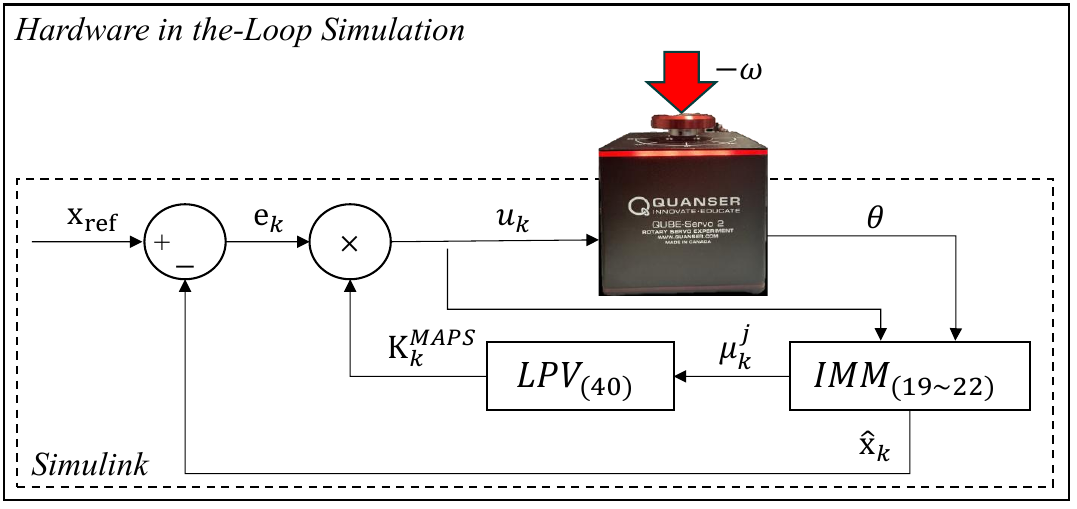}
  \caption{Experimental Architecture}
  \label{fig:friction}
\end{figure}

\section{Simulation Results and Analysis}
\label{sec:simulation}

This section presents simulation-based validation of the proposed MAPS framework under varying friction conditions. MATLAB simulations were performed to evaluate the estimation performance of the IMM-KF compared to standard KF. The friction coefficient $b$ was artificially varied between $b_{\mathrm{min}}$ and $b_{\mathrm{max}}$ to test adaptive estimation abilities.

\subsection{State Estimation Performance}

Fig.~\ref{fig:state_est}. compares the state estimation results of the IMM-KF and the standard KF for angular position ($\theta$), angular velocity ($\omega$), and current ($i$). In the position plot, both filters show similar accuracy during steady-state operation. however, clear differences arise during periods of friction variation. Notably, for unmeasured states such as velocity and current, the IMM-KF closely tracks the true state trajectory despite friction changes around $t = 0.2$--$0.4$~[s] in Fig.~\ref{fig:error}, while the standard KF exhibits delays and larger overshoot. This underlines the IMM-KF's superior ability to adapt to rapid, nonlinear parameter shifts. The IMM-KF's adaptability is especially evident in scenarios with parameter transitions, which are common in practical systems experiencing time-varying friction or load disturbances.

\begin{figure}[!h]
\centering
\includegraphics[width=0.48\textwidth]{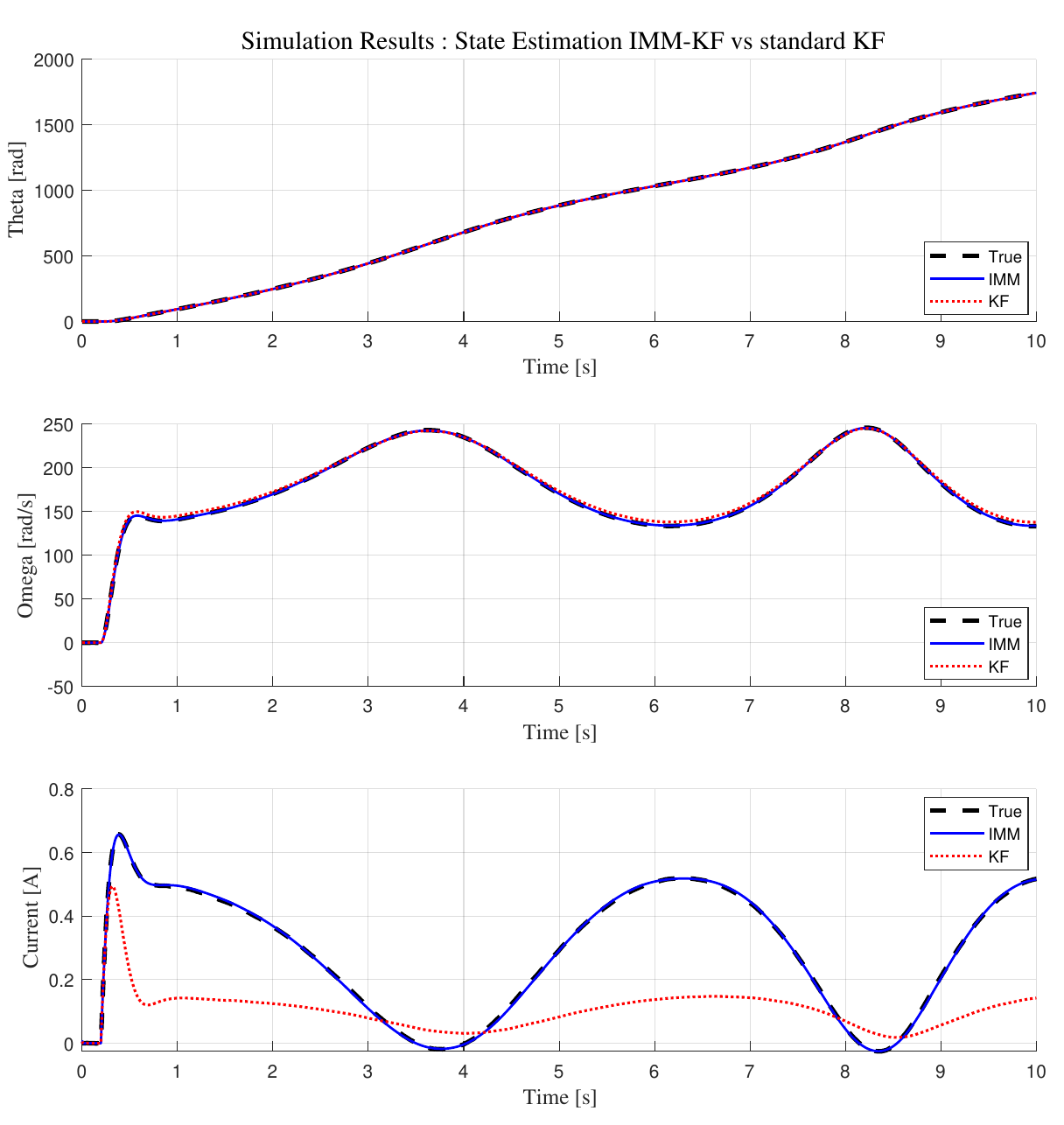}
\caption{State estimation comparison: IMM-KF(red) vs. standard KF(blue).}
\label{fig:state_est}
\end{figure}

\subsection{Estimation Error Analysis}

Fig.~\ref{fig:error}. illustrates the time-series estimation errors of angular position ($ \theta $), angular velocity ($ \omega $), and current ($ i $) for both the IMM-KF and the standard KF. While position estimates are comparable, the IMM-KF shows clear advantages in estimating unmeasured states specifically, angular velocity ($ \omega $) and current ($ i $). During abrupt friction transitions, the standard KF suffers from transient oscillations and steady-state error, whereas the MAPS maintains low error magnitudes and quickly reconverges after disturbances. This results in more consistent and reliable state feedback, crucial for control performance.Additionally, as shown in the mode probability subplot of Fig.~\ref{fig:error}, a rapid change in mode probability is detected around 0.4 seconds, indicating that the IMM-KF accurately recognizes and adapts to the dynamic mode transition. Correspondingly, the lower subplot illustrates the estimated friction coefficient ($b$), which closely tracks the true time-varying friction during gradual and slow changes. However, it is worth noting that the estimator exhibits a lagged response to abrupt shifts in friction, highlighting a limitation in rapidly capturing sudden parameter changes. These results clearly demonstrate that the MAPS framework enables robust and adaptive real-time mode identification, effectively handling slowly varying friction conditions while exhibiting some delay in response to abrupt transitions.These results jointly demonstrate that the MAPS not only improves traditional state estimation metrics, but also enables real-time identification of dynamic friction behavior, thereby strengthening the MAPS framework adaptability to real-world electro-mechanical systems.
\begin{figure}[!h]
\centering
\includegraphics[width=0.48\textwidth]{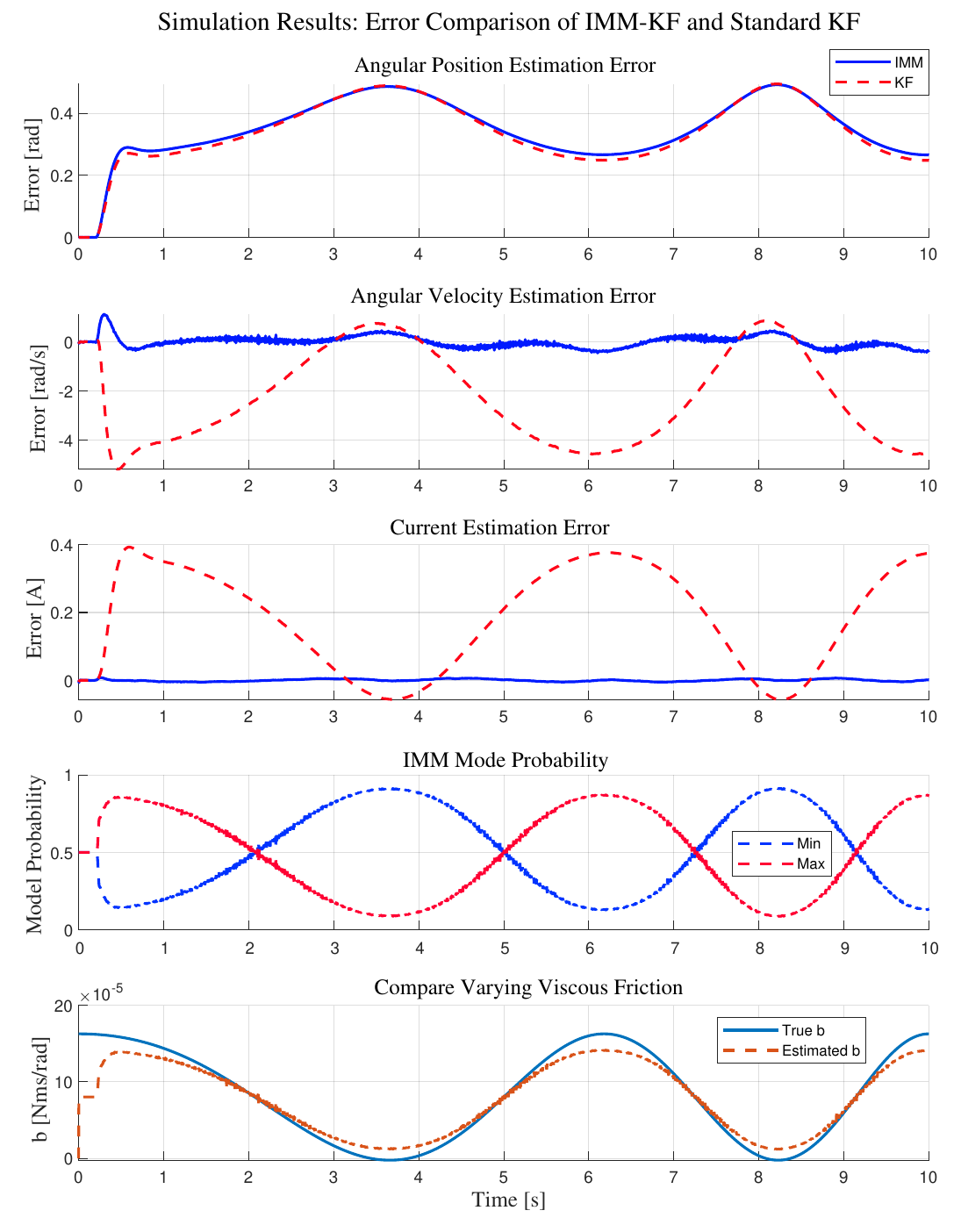}
\caption{State estimation error comparison: IMM-KF(red) vs. standard KF(blue).}
\label{fig:error}
\end{figure}

\subsection{Quantitative Comparison Using RMSE}
The quantitative performance of the MAPS was evaluated using the RMSE of the estimated states, as summarized in Table~\ref{tab:rmse_comparison}. This table clearly demonstrates that the MAPS achieves a notable reduction in RMSE for angular position ($\theta$), angular velocity ($\omega$), and current ($i$) estimation compared to the standard KF. While angular position ($\theta$) estimation errors remain nearly identical between the two methods, MAPS exhibits a significant advantage in estimating unmeasured states. Detailed values in Table~\ref{tab:rmse_comparison} reveal that MAPS reduces the angular velocity ($\omega$) estimation error by over 90\% and the current ($i$) estimation error by more than 96\%. These results indicate that the MAPS is highly effective at estimating unmeasured states, particularly under varying friction conditions. Accurate estimation of angular velocity and current is crucial for precise torque generation and robust friction compensation, both of which are essential for achieving high control performance and stability in real-time applications.

\begin{table}[!t]
\centering
\caption{RMSE Comparison of State Estimation}
\begin{tabular}{l|c|c}
\hline\hline
\textbf{State Variable} & \textbf{IMM-KF RMSE} & \textbf{Standard KF RMSE} \\
\hline
Angular Position ($\theta$) [rad] & 0.3698 & 0.3681 \\
Angular Velocity ($\omega$) [rad/s] & 0.1956 & 2.4313 \\
Current ($i$) [A] & 0.0091 & 0.2470 \\
\hline\hline
\end{tabular}
\label{tab:rmse_comparison}
\end{table}

\section{HILS Experimental Results and Analysis}
This section reports on the real-time HILS experiments designed to verify the effectiveness of the proposed MAPS framework in realistic, hardware-level environments. We systematically evaluate state estimation and control tracking performance under varying friction and load conditions.

\subsection{State Estimation Error Performance}
This section is to evaluate the state estimation performance of the IMM-KF compared to a standard KF under varying friction conditions. Experiments were conducted over a 30-second interval using a limited voltage range, and external friction was periodically applied and removed to emulate time-varying frictional behavior. The RMSE was adopted as the primary performance metric to quantify estimation accuracy.

\subsubsection{Visualization of Instantaneous Estimation Errors}

Fig.~\ref{fig:estimation_error_comparison}. illustrates the time-series plots of state estimation errors obtained from both the IMM-KF and the standard KF. The test scenario dynamically alternated between low and high friction states, which introduces nonlinearity and uncertainty to the plant behavior. The plotted error signals clearly demonstrate that the MAPS maintains a significantly tighter error bound across all states (angular position ($\theta$), angular velocity ($\omega$), and current ($i$)), especially during transitions caused by friction changes. In contrast, the standard KF exhibits noticeable estimation drifts and overshoots, indicating degraded performance under such uncertainties.

\begin{figure}[!ht]
    \centering
    \includegraphics[width=0.5\textwidth]{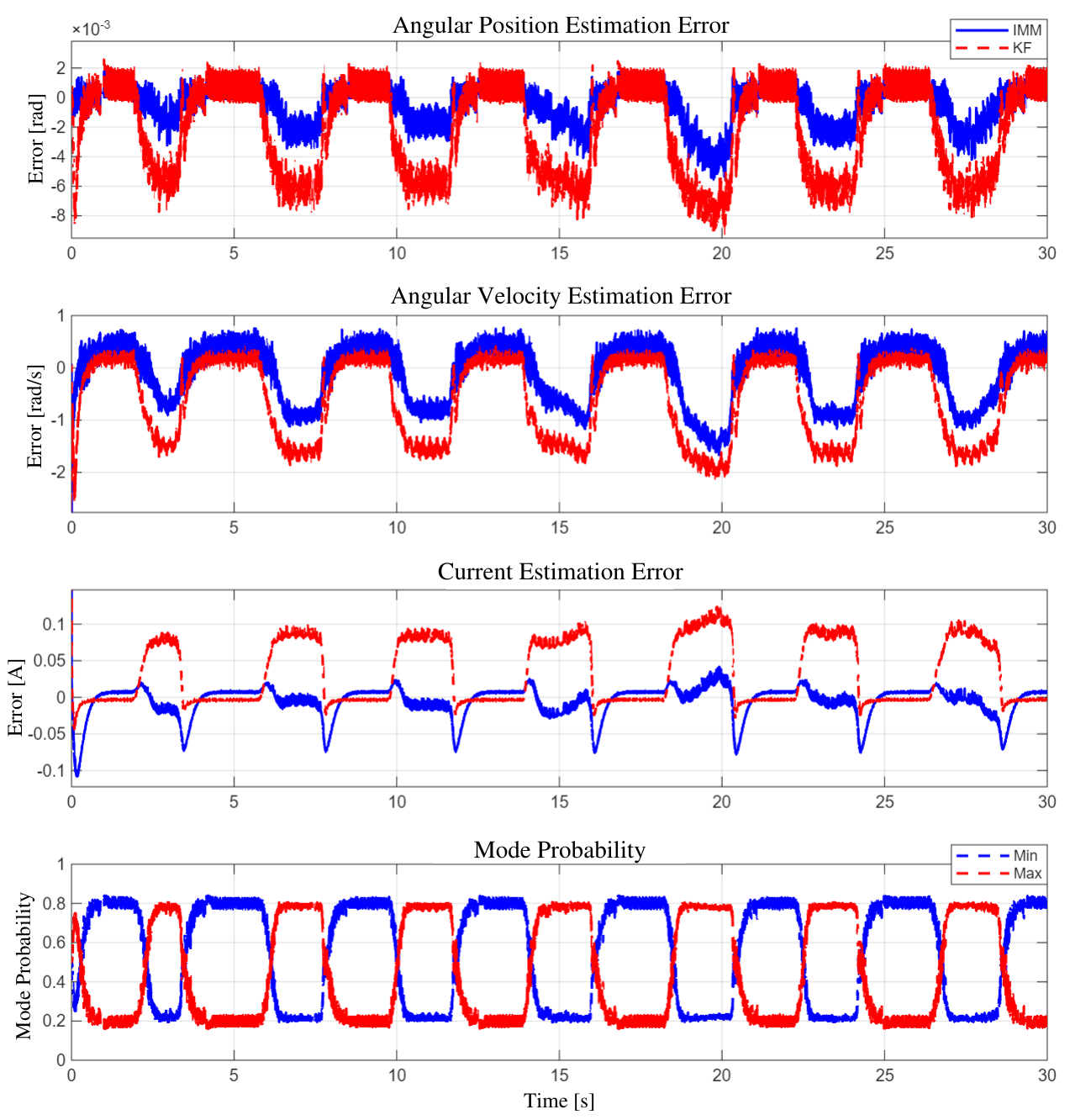}
    \caption{HILS-based state estimation error comparison: IMM-KF (red) vs. standard KF (blue).}
    \label{fig:estimation_error_comparison}
\end{figure}

\subsubsection{Quantitative Comparison Using RMSE}
To provide a numerical comparison, Table~\ref{tab:hils_rmse_comparison} summarizes the RMSE values for each state variable. The IMM-KF demonstrates substantial improvements across all states. In particular, the estimation error for angular position (\( \theta \)) is reduced by approximately 63\%, while angular velocity (\( \omega \)) and current (\( i \)) RMSE decrease by about 42\% and 57\%, respectively. These results quantitatively confirm the superior estimation accuracy of the IMM-KF in time-varying and uncertain control scenarios.

\begin{table}[!t]
\centering
\caption{HILS-based RMSE Comparison of State Estimation}
\begin{tabular}{l|c|c}
\hline\hline 
\textbf{State Variable} & \textbf{MAPS RMSE} & \textbf{Standard KF RMSE} \\
\hline
Angular Position ($\theta$) [rad] & 0.0014 & 0.0038 \\
Angular Velocity ($\omega$) [rad/s] & 0.5725 & 0.9825 \\
Current ($i$) [A]                 & 0.0223 & 0.0523 \\
\hline\hline
\end{tabular}
\label{tab:hils_rmse_comparison}
\end{table}

These findings highlight the MAPS robustness in coping with system uncertainties and its adaptive capability for accurate state tracking, demonstrating a clear advantage over the standard KF architecture in HILS-based evaluation settings.


\subsection{Control Performance: MAPS-gain LQR vs. fixed-gain LQR Controller}

To compare control performance, both MAPS-gain LQR controller and fixed-gain LQR controller were tested under identical reference and friction conditions.
To comprehensively evaluate the controller performance, both step response and reference tracking tasks were conducted using various reference signals. Specifically, square and sine waveforms were employed as reference inputs, with amplitudes ranging from $-4\,\mathrm{V} \; \text{to} \; 4\,\mathrm{V}$. The amplitudes and frequencies of these signals were carefully chosen to remain within the physical limits of the QUBE-Servo 2 platform, thereby ensuring safe operation throughout all experiments.

\subsection{Step Response without external load}

\begin{figure}[!ht]
    \centering
    \includegraphics[width=0.5\textwidth]{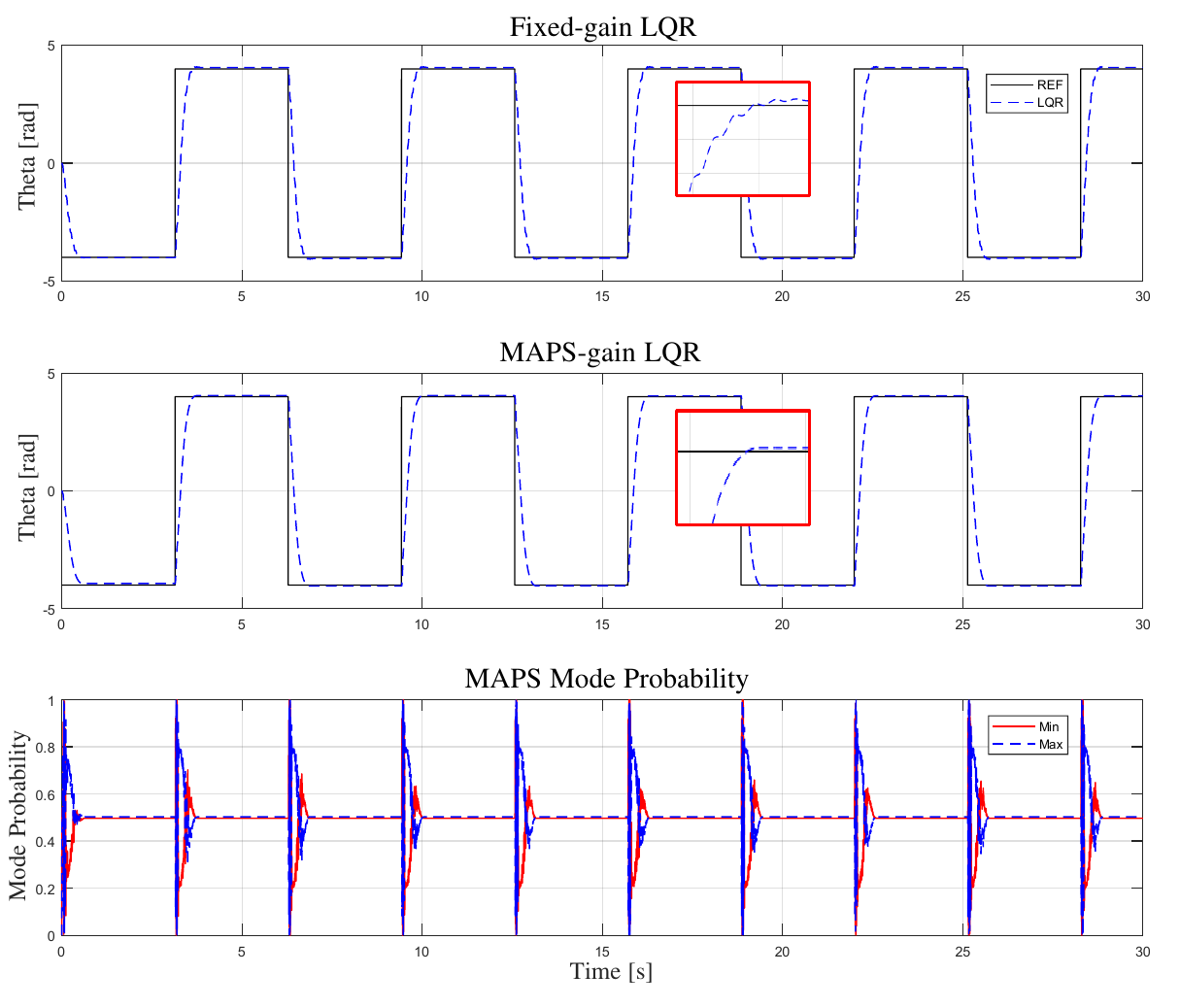}
    \caption{Closed-loop step response without external load: comparison between fixed-gain LQR controller and MAPS-gain LQR controller.}
    \label{fig:step_notouch}
\end{figure}

Fig.~\ref{fig:step_notouch}. presents the closed-loop step responses under nominal conditions, i.e., without external load. Both the fixed-gain LQR controller and the MAPS-gain LQR controller successfully track the reference signal with minimal steady-state error. However, the MAPS-gain LQR controller demonstrates significantly better transient performance, with notably smaller overshoot and faster settling time. This is primarily due to its adaptive gain adjustment in response to real-time mode probability, as shown in the lower subplots. The mode probability of the MAPS transitions smoothly with system dynamics, indicating its capacity to accommodate subtle nonlinear behaviors even without external load.

\subsection{Step Response with external load}

\begin{figure}[!ht]
    \centering
    \includegraphics[width=0.5\textwidth]{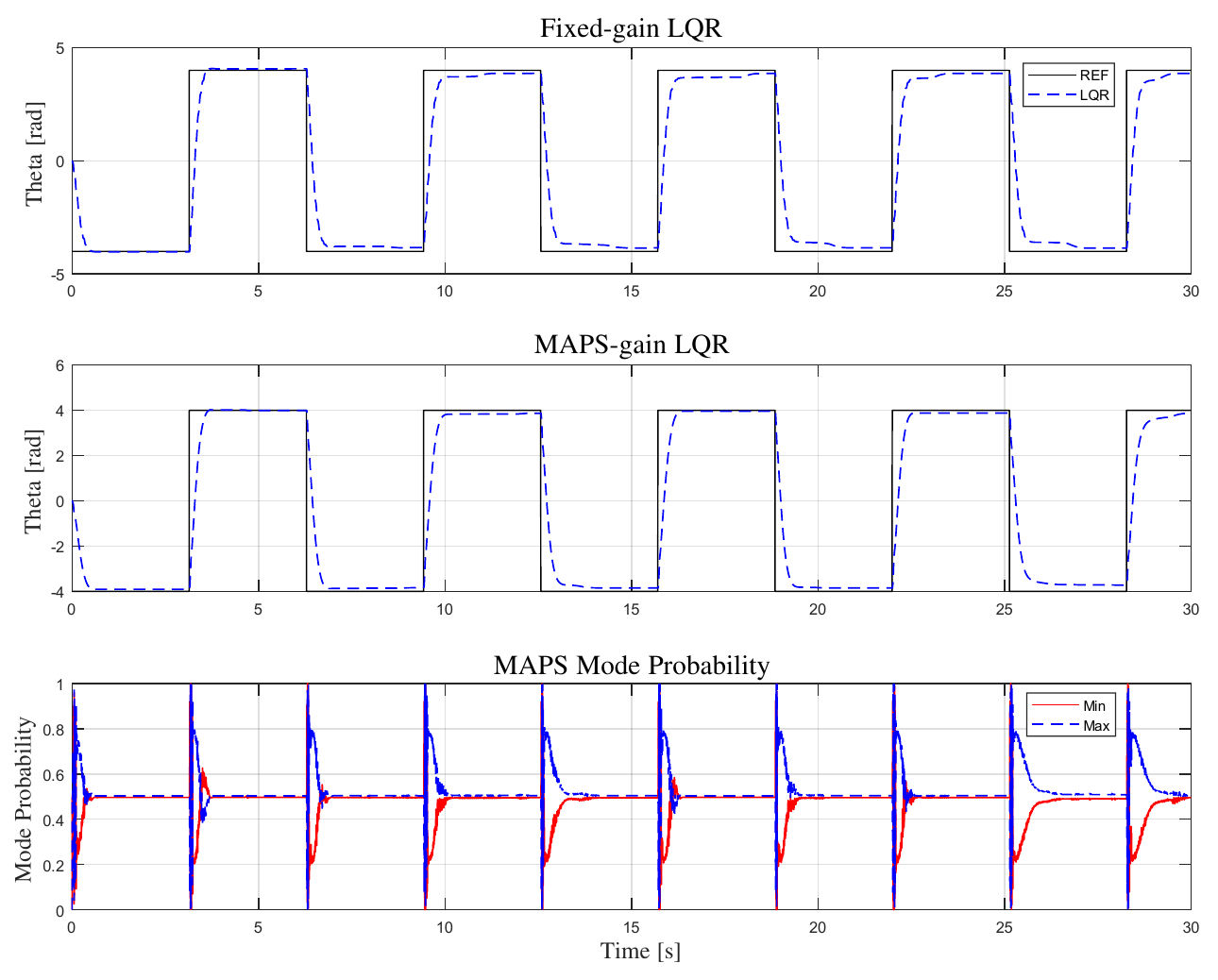}
    \caption{Closed-loop step response with external load: comparison between fixed-gain LQR controller and MAPS-gain LQR controller.}
    \label{fig:step_touch}
\end{figure}

Fig.~\ref{fig:step_touch}. illustrates the system response when an external load is introduced via external load during step tracking. Under this perturbation, the fixed-gain LQR controller exhibits a noticeable degradation in tracking performance, with prolonged transient oscillations and significant steady-state offset. In contrast, the MAPS-gain adapts its LQR controller gain based on changing mode probability, which reflects a friction-induced dynamic mode shift. As a result, its tracking response remains well-regulated, with considerably shorter settling time and smaller error. The effectiveness of mode scheduling in real-time gain adaptation is clearly evidenced here.

\subsection{Sine Tracking without external load}

\begin{figure}[!ht]
    \centering
    \includegraphics[width=0.5\textwidth]{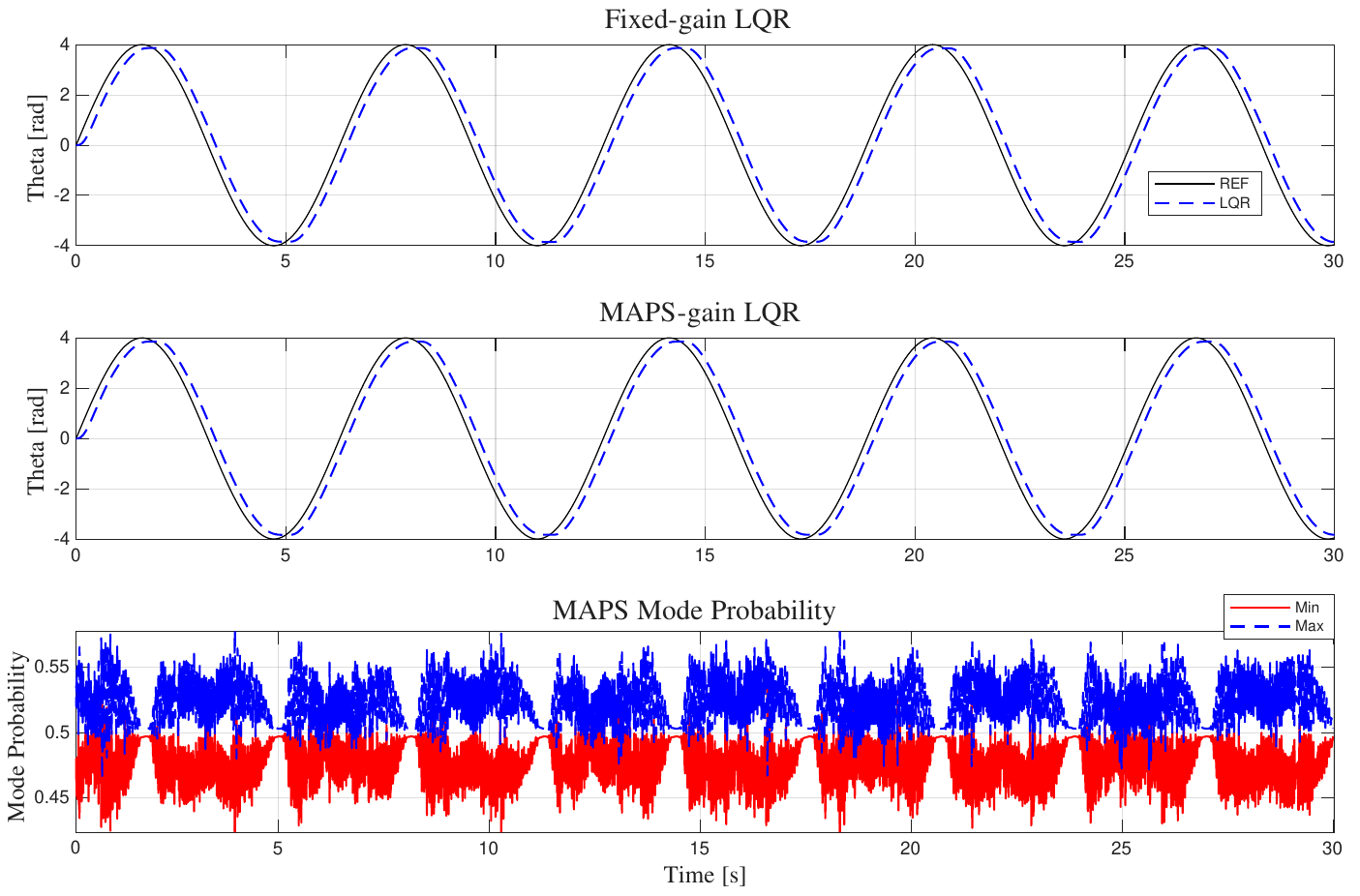}
    \caption{Closed-loop sine tracking responses without external load: comparison between fixed-gain LQR controller and MAPS-gain LQR controller.}
    \label{fig:sine_notouch}
\end{figure}

Fig.~\ref{fig:sine_notouch}. shows the tracking performance for a sinusoidal reference signal in the absence of external load. In this nominal condition, both controllers follow the reference trajectory closely. However, the MAPS-gain LQR controller exhibits superior phase alignment and reduced amplitude distortion throughout the input cycle. The mode probability graph indicates a low activation of alternative modes, implying that the MAPS controller maintains a stable regime under smooth operating conditions. The consistent performance of the MAPS-gain LQR controller in this scenario emphasizes its robustness without sacrificing tracking precision.

\subsection{Sine Tracking with external load}

\begin{figure}[!ht]
    \centering
    \includegraphics[width=0.5\textwidth]{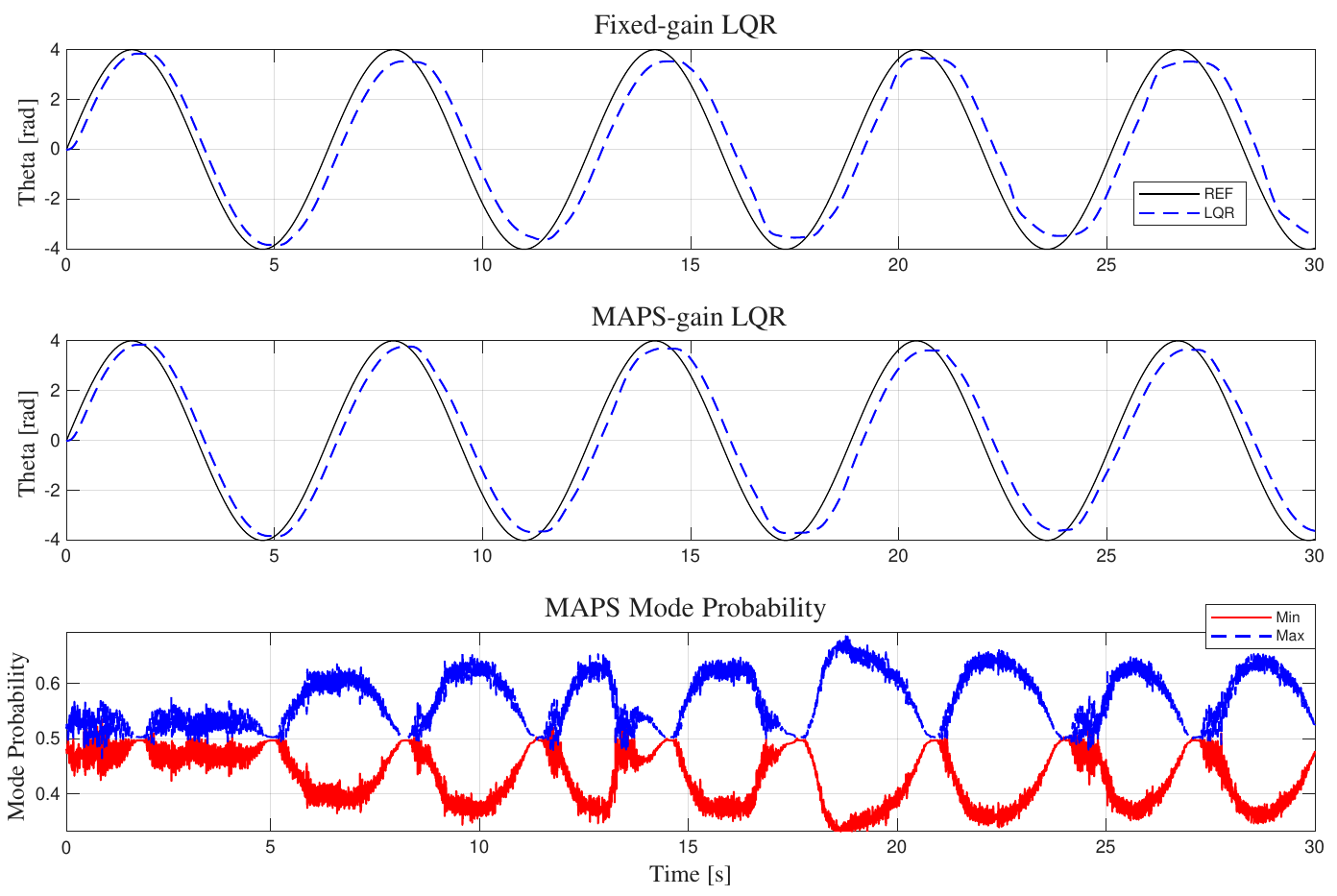}
    \caption{Closed-loop sine tracking responses with external load: comparison between fixed-gain LQR controller and MAPS-gain LQR controller.}
    \label{fig:sine_touch}
\end{figure}

Finally, Fig.~\ref{fig:sine_touch}. presents the sine wave tracking under external load. The fixed-gain LQR controller struggle to maintain phase and amplitude response consistency, displaying tracking lag and steady-state deviations across the waveform. The MAPS-gain LQR controller, on the other hand, dynamically adjusts its gain in response to the friction disturbance evidenced by significantly shifting mode probabilities and successfully preserves accuracy in oscillatory tracking. These results underline the MAPS-gain LQR controller clear advantage in environments with time-varying uncertainties such as friction or load shifts.

\subsection{Control performance Summary}
As summarized in Table~\ref{tab:tracking_comp}, the performance comparison between the MAPS-gain LQR and fixed-gain LQR controllers reveals several key insights. Under nominal step input conditions without disturbances, both controllers achieved similar tracking performance, with the fixed-gain LQR slightly outperforming MAPS-gain LQR by a margin of 1.67\% in IAE. However, as the control environment becomes more dynamic—particularly during sine tracking and under external brake disturbances—the MAPS-gain LQR demonstrates clear advantages. In the presence of external friction during sine tracking, MAPS-gain LQR reduced RMSE by 16.11\%, MAE by 15.24\%, and IAE by 15.24\% compared to the fixed-gain LQR. These improvements confirm the MAPS controller ability to adapt to mode transitions and retain accuracy in uncertain, time-varying environments.

The use of RMSE, MAE, and IAE as performance metrics offers a comprehensive evaluation of tracking quality. RMSE captures the average magnitude of error with stronger sensitivity to outliers, thus indicating transient performance. MAE provides a robust measure of the average tracking deviation, reflecting overall response consistency. IAE quantifies the cumulative effect of tracking errors over time, serving as a strong indicator of long-term control precision. Together, these metrics allow for a balanced and multi-dimensional assessment of both immediate and accumulated control performance~\cite{yuan2019error, chen2008best}.

\begin{table}[!ht]
\centering
\caption{Controller Performance Comparison: \\ Fixed-gain LQR vs. MAPS-gain LQR}
\begin{tabular}{c|c|c|c}
\hline\hline
\textbf{Scenario} & \textbf{Metric} & \textbf{Fixed-Gain} & \textbf{MAPS-Gain} \\
\hline
\multirow{3}{*}{\makecell{Step \\ (w/o load)}} & RMSE & 1.4756 & 1.4817  \\
                                               & MAE  & 0.4735 & 0.4815  \\
                                               & IAE  & 14.2026 & 14.4405 \\
\hline
\multirow{3}{*}{\makecell{Sine \\ (w/o load)}} & RMSE & 0.5348 & 0.5110 \\
                                               & MAE  & 0.4920 & 0.4715 \\
                                               & IAE  & 14.7622 & 14.1458 \\
\hline
\multirow{3}{*}{\makecell{Step \\ (w/ Load)}}  & RMSE & 1.5001 & 1.4915 \\
                                               & MAE  & 0.6243 & 0.5886  \\
                                               & IAE  & 18.7269 & 17.6559 \\
\hline
\multirow{3}{*}{\makecell{Sine \\ (w/ Load)}}  & RMSE & 0.8356 & 0.7009 \\
                                               & MAE  & 0.7485 & 0.6344 \\
                                               & IAE  & 22.4560 & 19.0336 \\
\hline\hline
\end{tabular}
\label{tab:tracking_comp}
\end{table}

\section{Conclusion}
In this paper, we proposed a novel adaptive control framework MAPS, which integrates an IMM estimator with a LPV control strategy. The key innovation in MAPS is leveraging real-time mode probabilities from the IMM as convex scheduling weights for gain interpolation in LPV-based control. Unlike conventional methods that depend on measurable physical scheduling variables, MAPS utilizes probabilistic inference to dynamically adapt control gains, enabling robust and effective management of time-varying friction and disturbances.

The practical applicability and performance improvements of the MAPS framework were extensively validated through HILS experiments on the QUBE-Servo 2 platform. Compared to a standard KF, MAPS achieved significant reductions in estimation errors up to 63\% in angular position (\( \theta \)), 42\% in angular velocity (\( \omega \)), and 57\% in current (\( i \)). Furthermore, the gain-scheduled LQR controller within the MAPS framework reduced closed-loop tracking RMSE by 16.11\% relative to a fixed-gain LQR controller across various reference tracking tasks under friction variation.

From a theoretical standpoint, the MAPS controller guarantees closed-loop quadratic stability by maintaining the controller gain as a convex combination of stabilizing vertex gains under a common Lyapunov function. This stability assurance underscores the robustness of the MAPS framework when confronted with uncertainties inherent in real-world systems.

Looking forward, future research will extend MAPS to more general control settings, including robust control frameworks such as \(\mathcal{H}_2\) and \(\mathcal{H}_\infty\) controllers, to investigate its effectiveness under scenarios where controllability or observability may be compromised. Additionally, we aim to explore the integration of dynamic system models with artificial intelligence techniques, broaden applications to multi-modal vehicular platforms, and implement MAPS on full-scale embedded control architectures.

By uniting probabilistic mode estimation with optimal gain scheduling, MAPS lays a promising foundation for scalable and generalizable adaptive control solutions in next-generation intelligent mobility systems.

 
\bibliographystyle{Bibliography/IEEEtranTIE}
\bibliography{references}


\end{document}